\documentclass[american,aps,pra,reprint,superscriptaddress,longbibliography]{revtex4-1}
\usepackage[unicode=true,pdfusetitle, bookmarks=true,bookmarksnumbered=false,bookmarksopen=false, breaklinks=false,pdfborder={0 0 0},backref=false,colorlinks=false] {hyperref}
\hypersetup{colorlinks,linkcolor=myurlcolor,citecolor=myurlcolor,urlcolor=myurlcolor}
\usepackage{braket,colortbl,amsthm,amsmath,amssymb,txfonts,graphicx}
\definecolor{myurlcolor}{rgb}{0,0,0.7}
\usepackage{cleveref}
\usepackage{soul}
\usepackage{subfigure}

\theoremstyle{plain}

\usepackage{ragged2e}  
\usepackage{mathrsfs}
\usepackage{physics}
\usepackage{mwe}
\usepackage{pifont}

\DeclareMathAlphabet{\mathbcal}{OMS}{cmsy}{b}{n}
\DeclareMathOperator{\sgn}{sgn}
\usepackage{soul}
\usepackage{pifont}
\newtheorem{customthm}{Theorem}

\usepackage{tikz}

\begin{document}

\title{Collective purification of interacting quantum networks beyond symmetry constraints}

\author{Saikat Sur}
\affiliation{Department of Chemical and Biological Physics \& AMOS,
Weizmann Institute of Science, Rehovot 7610001, Israel}
\email{saikats@imsc.res.in (corresponding author)}
\affiliation{Optics \& Quantum Information Group, The Institute of Mathematical Sciences, HBNI, CIT Campus, Taramani, Chennai 600113, India}

\author{Pritam Chattopadhyay}
\affiliation{Department of Chemical and Biological Physics \& AMOS,
Weizmann Institute of Science, Rehovot 7610001, Israel}
\author{Arnab Chakrabarti}
 \affiliation{Department of Physics, Rajiv Gandhi University,
 Rono Hills, Doimukh - 791112, Arunachal Pradesh, India}
\affiliation{Department of Chemical and Biological Physics \& AMOS,
Weizmann Institute of Science, Rehovot 7610001, Israel}

\author{Nikolaos E. Palaiodimopoulos}
\affiliation{Institute of Electronic Structure and Laser, FORTH, GR-70013 Heraklion, Crete, Greece}

\author{\"{O}zg\"{u}r E.~M\"{u}stecapl{\i}o\u{g}lu}
\affiliation{Ko\c{c} University, Department of Physics, Sar{\i}yer, Istanbul, 34450, T\"urkiye} 
\affiliation{TÜBİTAK Research Institute for Fundamental Sciences (TBAE), 41470 Gebze, T\"urkiye}

\author{Amit Finkler}
\affiliation{Department of Chemical and Biological Physics \& AMOS,
Weizmann Institute of Science, Rehovot 7610001, Israel}

\author{Durga Bhaktavatsala Rao Dasari}
\affiliation{3rd Institute of Physics, University of Stuttgart, ZAQuant, Stuttgart, Germany}

\author{Gershon Kurizki}
\email{gershon.kurizki@weizmann.ac.il}
\affiliation{Department of Chemical and Biological Physics \& AMOS,
Weizmann Institute of Science, Rehovot 7610001, Israel}
\date{\today}
\begin{abstract}
\begin{center}
    \textbf{ABSTRACT}
\end{center}
Following any quantum information processing protocol, it is essential to reset a mixed state of a many-body \textit{interacting} spin-network to the computational-zero pure state. This task is challenging, both theoretically and experimentally, because of the quantum correlations. There is currently no effective cooling strategy for both high and low temperatures in such networks. Here we put forth a universal cooling strategy for multi-spin interacting networks. The strategy is based on the \textit{collective coupling} of the system to an ancilla spin that intermittently dumps part of its entropy into an ultracold bath. Yet this strategy should overcome the \textit{symmetry-imposed} correlations that impede the cooling. To avoid the prohibitive complexity of computing the dynamics, we resort to graph analysis of the network. 
We show that a \textit{unique choice} of alternating, non-commuting system-ancilla interaction Hamiltonians exists that breaks the symmetry constraints and allows the network to approach the desired pure state. We illustrate this universal purification strategy in diverse experimental settings.
\end{abstract}

\maketitle

\section{Introduction}
 Quantum information processing (QIP), be it for quantum computing~\cite{loss_1998,benjamin2003quantum}, quantum simulation~\cite{yamamoto2017coherent,johnson2011quantum,king2022coherent,georgescu_2014}, or quantum communication~\cite{wang,kimble2008,bose}, relies on many-body entanglement generation and distribution in the system (alias network)~\cite{subra1,manmana,chiara,zwick_2014,PhysRevResearch.5.033214}. All such tasks require controlled interactions between individual constituents of the system, commonly qubits (spin-$1/2$ systems). In order to achieve high fidelity, the entangling interactions must favorably compete with decoherence~\cite{kurizki_2015_pnas,petruccione,kurizki_kofman}. Yet the decoherence challenge is not the only one faced by QIP: once the task has been completed, the many-body state must be disentangled or read out and reset/cooled down to the initial, preferably pure, state prior to the commencement of the next task~\cite{miguel_2018,oftelie_prxq_2024}. This undertaking may be as challenging, because the many-body state to be reset is at least partially mixed while the interactions and/or correlations between the constituents persist~\cite{manmana,chiara,zwick_2014}. 

One might be tempted to contemplate a passive network-resetting approach: simply wait for the network to cool off by putting it in contact with a very cold bath, say, by radiative cooling. However, such cooling produces after few $T_1$ periods a low-temperature thermal distribution of eigenstates, whereas exclusively populating the ground state (just like any other pure state) is prohibited by the Third Law~\cite{nernst_2}. Even achieving close proximity to the ground state by passive cooling is practically prohibitive in many systems and not well-suited for quantum technologies:  
Passively waiting for an approximate reset 
in solid-state nuclear spin ensembles coupled to e.g., nitrogen vacancy (NV) impurities embedded in diamond~\cite{Schirhagl2014_NV_review} takes incredibly long and is thus not an option. Even in superconducting qubits, 
a near-ground-state reset through passive decay would take several milliseconds \cite{Bland2025}, orders of magnitude longer than typical gate or measurement times. In scalable quantum processors or quantum networks operating via repeated measurements or feedback, such passive resetting periods drastically limit throughput, stability, and temporal synchronization of qubits~\cite{Chae2024_qubit_init_review, Krantz2019_quantum_engineer_guide}: since, according to the Third Law, passive relaxation is a stochastic process, it does not guarantee a return to the ground state, resulting instead in residual excitations that accumulate as initialization errors over repeated runs.

{\color{black} Hence, active reset protocols that may enable deterministic and rapid purification of multi-qubit quantum states are essential prerequisites for scalable and high-fidelity QIP operations involving quantum error correction~\cite{devitt2013quantum} and algorithmic resets~\cite{boykin_pnas_2002}.}

Active purification (alias polarization) is acutely needed also in permanently correlated multi-spin systems, such as large molecules that are thermalized by their environment and therefore have to be ``hyper-polarized" in order to yield coherent signals, as required in NMR or MRI~\cite{fel1997multiple,bonizzoni2024quantum}. 
NMR literature~\cite{emsley_book, bernstein_1957,feng_2012,kuzyk_2006}  has investigated the connection of molecular symmetry to hyperpolarization, mainly experimentally. However, to our knowledge, a systematic theoretical analysis of the interplay between symmetry, topology, and polarization or cooling in general spin networks is still lacking, and the governing principles are not yet fully understood. As we show here, traditional arguments based on permutation-group representations~\cite{levit_2016} may not be sufficient to determine polarization bounds in general spin networks, where additional and more subtle structural factors can become relevant.

Despite the major importance of the problem, a comprehensive and broadly applicable strategy for addressing the general question remains to be established: How does one actively cool down an interacting many-body quantum system (network) close to a pure state, and what determines the efficacy and speed of such cooling?

Theoretically, the purification/or cooling of interacting multi-spin systems is mostly uncharted terrain. Instead, theory has primarily dealt with the dissipative cooling of \textit{non-interacting} spin ensembles~\cite{njp_2020,abiuso_qst_2024} or single qubits~\cite{tal_mor_pra_2016,boykin_pnas_2002}. 
The dynamics of such processes has commonly been simplified by semiclassical~\cite{rousochatzakis_2018} or master-equation~\cite{petruccione} approaches which do not keep track of many-body quantum correlations despite their crucial role, as shown here. Consequently, such approaches fail as zero temperature is approached~\cite{kolar_2012}.   Experimentally, cooling of multi-spin systems~\cite{tan2017quantum} has had only partial success~\cite{zwick_2014} or faced the obstacle, as in nuclear spin ensembles, that these spins cannot be directly cooled since they are decoupled from external phonon or radiation baths.  

Here we put forth a general strategy based on \textit{collective cooling} or purification of interacting $N$-spin systems (networks)~\cite{villazon_2020} via a single ancillary qubit~\cite{schlipf_2017} that intermittently couples to and decouples from the system and repeatedly dumps part of the system collective entropy into a very cold (e.g., radiative) bath~\cite{durga_et_al}. The ancilla cooling to its ground state completes the cycle, which then recommences with the ancilla-system recoupling. The cycles constitute recursive dynamical maps, fully allowing for quantum correlations. The goal is to bring the system towards the computational zero state $|000\cdots 0\rangle$, alias the ferromagnetic ground state (FGS) $|\downarrow \downarrow \downarrow \cdots \downarrow\rangle$, as rapidly as possible.

Notwithstanding its apparently universal and minimalistic character, such a purification strategy faces a principal difficulty inherent in interacting multi-spin systems, namely, \textit{quantum correlations} among the spins that arise if the system has any spatial or spectral symmetry. As argued here, these \textit{symmetry-imposed quantum correlations severely hinder the system purification by the ancilla: they prevent equilibration of the system}, and thus invalidate the eigenstate thermalization hypothesis (ETH) which is presumed universal~\cite{deutsch_eth_2018,rigol2008thermalization,cramer2008exact}. We address this difficulty by posing and answering two complementary fundamental questions:

A) \textit{How does the achievable purification (polarization) of the $N$-spin network, measured by its state purity $\text{Tr} (\rho^2_N)$, depend on the network+ancilla geometry and topology}?
As has been shown for \textit{non-interacting} multi-spin networks~\cite{njp_2020}, \textit{collective ancilla-system coupling} may cause the system spins to evolve into a thermal mixture of \textit{conserved multi-spin observables}, which hinders their polarization.   Yet, it is unclear how (if at all) can such considerations be extended to \textit{interacting spin networks}. This requires the \textit{diagonalization} of their many-body Hamiltonian, which is \textit{exponentially} time-consuming as a function of the network size, and thus appears to pose an insurmountable obstacle~\cite{balasubramanian_1995}. 

To bypass this obstacle,  we here invoke \textit{graph theory} to infer the \textit{symmetry-limited}~\cite{cao_2} steady-state polarization values of the \textit{network polarizability} $\text{Tr}(\rho^2_N)$. This approach is shown to enormously simplify the calculations: the exponentially difficult solutions of the quantum evolution of arbitrarily interacting spin networks that possess symmetries are mapped onto solutions of much lower polynomial complexity, which yield simple bounds of the network polarizability. Graph theory has been previously employed to find the ground-state properties of large-spin systems with general topologies in the thermodynamic limit~\cite{hirjibehedin2006spin,khajetoorians2012atom,spinelli2014imaging,searle_2024}. Here, graph theory reveals symmetry-imposed polarizability bottlenecks in complex spin networks of \textit{any size}, with either short-or long-ranged (non-nearest neighbor) interactions~\cite{kurizki_kofman, shahmoon_2016} that are described by graphs with any connectivity,  not only bipartite graphs~\cite{ross_2010}. 

B) \textit{Can the symmetry bottlenecks be overcome, and if so, how?  }
The answer we give is positive: we propose a universal strategy that successfully overcomes the symmetry bottlenecks by resorting to \textit{uniquely chosen}, alternating, \textit{non-commuting} system-ancilla interaction Hamiltonians. This procedure, as we prove, leads to complete purification.

In what follows, we first outline the collective purification hindrances due to symmetry (Sec. II A-C), then propose a general strategy to bypass them, leading the system to the FGS (Sec. IID). We finally illustrate the general strategy for analytically solvable cases of the generic Heisenberg spin chains where purification is shown to approach the FGS (Sec. IIE) and examine experimentally feasible platforms (Sec. IIF). The results and the outlook are discussed in Sec. III.

\begin{figure*}
   \includegraphics[width=0.40\textwidth]{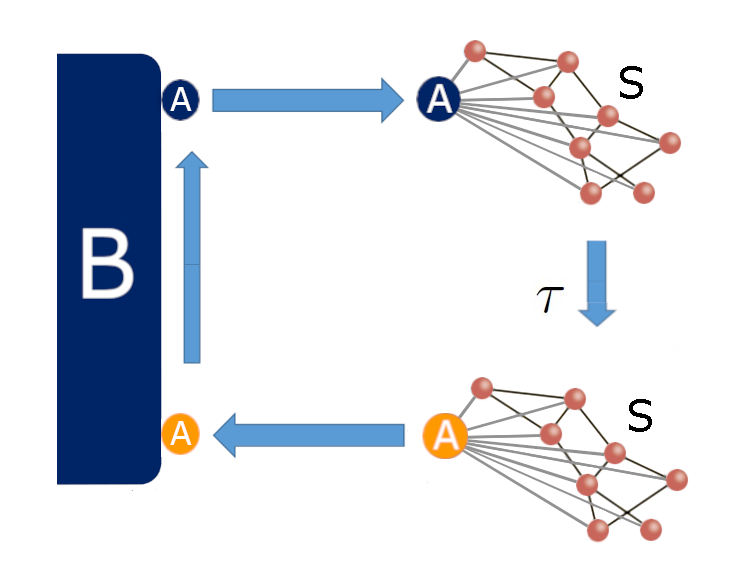}

  \caption{ \textbf{Purification Protocol:} Schematic diagram of the purification protocol: Interacting spin network cooling/purification via collective swapping of the network ($S$) entropy with an ancilla qubit ($A$) in recurring cycles. The ancilla is intermittently reset/purified by an ultracold (ideally, zero-temperature) bath $(B)$.}
  \label{fig:circuit} 
  \end{figure*}

\section{Results}

\subsection{Scenario}

Consider a system ($S$) formed by an arbitrarily interacting network of spin-$1/2$ objects that is coupled to an externally controlled ancilla ($A$). The setup is described by the Hamiltonian

\begin{equation}
\hat H(t) =  \hat H_{S} +  \hat H_{A}(t) +  \hat H_{SA}(t),  
\end{equation}
The system Hamiltonian $S$ has the highly general form
\begin{equation}
\hat H_{S} = \sum_{i,j} \mathbcal{J}_{ij}(\hat \sigma_i^x\hat \sigma_{j}^x+\hat \sigma_i^y\hat \sigma_{j}^y + \Delta \hat \sigma_i^z\hat \sigma_{j}^z), 
\label{network_hamiltonian}
 \end{equation}
where $\mathbcal{J}_{ij}$ denote the dipole-dipole couplings  between any two $(i,j)$ spins $i,j=1,2,...,N$, $\sigma^\alpha_i (\alpha 
 = x,y,z)$ being their Pauli spin operators, and $\Delta$ is the network anisotropy parameter corresponding to the bound state energy of two down (up) spins. This Hamiltonian allows for non-nearest neighbor interactions~\cite{thiru,shahmoon_2016,nick,defenu2023long,Pritam2024QST}. The network geometry may be spatially periodic or aperiodic in one, two, or three dimensions, with any connectivity/topology.

The ancilla (qubit) Hamiltonian is given by
\begin{equation}
\hat H_{A}(t) = h_{A}(t)\hat\sigma^z_A,
\label{eq: probe_Hamiltonian}
\end{equation}
with  time-dependent (modulated) two-level splitting $h_A(t)$, and $A$ is coupled to the $S$ spin sites $k$ via 
 \begin{equation}
\hat{H}_{SA}(t) = \frac{1}{2} \sum_k g_k(t)(\hat{\sigma}^+_A \hat{\sigma}^-_k + \hat{\sigma}^-_A \hat{\sigma}^+_k).
\label{eq: probe-spin-hamiltonian}
\end{equation}
We shall consider different $k$-site distributions of the $S-A$ couplings $g_k(t)$.

 \subsection{Steady-state purification Maps}

The network $S$ is initially in a high-temperature (nearly fully-mixed) state. A simple purification protocol~\cite{njp_2020,raghunandan_2020, villazon_2020}  is cyclic with period $\tau$. Each cycle starts with the $A$ spin initialized in the state $\vert 0 \rangle_A$, $\hat{H}_{SA}$ is quenched as the couplings $g_k(t)$ are switched on, correlating $A$ and $S$, while $\hat{H}_S$ acts constantly. Then $\hat{H}_{SA}$ is quenched again by switching off $g_k(t)$ to $0$. The spin $A$ is subsequently coupled to a cold bath where it dumps the entropy acquired during the $A-S$ interaction, so that $A$ is reset to $\vert 0 \rangle_A$ and the cycle recommences (Fig.~\ref{fig:circuit}).


The recursive dynamics is then
\begin{eqnarray}
	\rho_S^{(n+1)} = \Tr_A  (\hat{U}(\tau) \ket{0}_{AA}\bra{0}\otimes \rho_S^{(n)} \hat{U}^{\dagger}(\tau)). 
	\label{eq: map_n}
\end{eqnarray}
Here, $\rho_S^{(n)}$ and $\rho_S^{(n+1)}$ are the network  states before and after the $n$'th reset, respectively, and  the time-ordered evolution operator is $\hat{U}(\tau) = T_{\rightarrow} \text{e}^{-i\int^\tau_0 \hat{H}(t) dt}$.

To find the steady-state solution (SSS), one should diagonalize the Hamiltonian $\hat H_S$ and compute the evolution towards the SSS. However, the time complexity of this diagonalization and the time evolution is exponential in $N$, the system size, which may be prohibitively large, $O(4^{3N})$, so that for networks with more than a few spins this task is unfeasible with conventional computers.



  
The quest for feasible alternatives may exploit the fact that the  SSS of  Eq.~(\ref{eq: map_n}) are generated by the linear map (SI I)
\begin{eqnarray}
 \mathbcal{M} \rho^{(n\rightarrow \infty)}_S  \rightarrow 0.
 \label{linear_map}
\end{eqnarray}
that constitutes a set of linear homogeneous equations for the elements of  $\rho_S$. Here $\mathbcal{M} \rho^{(n\rightarrow \infty)}_S = \Tr_A  (\hat{U}(\tau) \rho_A \otimes \rho^{(n\rightarrow \infty)}_S \hat{U}^\dagger(\tau) - \rho^{(n\rightarrow \infty)}_S)$. We seek for any $\tau$  (unlike $\tau$-periodic maps) conditions for convergence to the fully-polarized $\vert 000...0\rangle$ state of $\hat{H}_S$, which is a fixed point (SSS) of the map $\mathbcal{M}$. We resort to two simplifications: 

(a) We show that the network polarizability for any $\hat{H}_S$ is determined by the SSS multiplicity, which can be inferred from corollaries of the Rouché-Capelli (RC) theorem~\cite{capelli_1892}:\\

\noindent RC Corollary 1. The FGS is the only SSS of $\mathbcal{M}$ if the rank $\mathscr{R}$ of $\mathbcal{M}$ satisfies
$\mathscr{R}(\mathbcal{M}) = \mathscr{D}(\mathbcal{M})$, where $\mathscr{D}(\mathbcal{M}) =  4^N - 1$ is the $N$-spin Hilbert-space dimensionality. This is a sufficient condition for a spin network to become fully polarized irrespective of the initial state. \\  
    
\noindent RC Corollary 2. $\mathbcal{M}$ has an infinite number of SSS, including the FGS, if $\mathscr{R}(\mathbcal{M}) < \mathscr{D}(\mathbcal{M})$. Then, the spin network can be fully polarized only for certain initial states.   

(b)  The number of independent solutions (nullity) $\mathscr{N}(\mathbcal{M})$  that span the SSS subspace is given by
\begin{eqnarray}
\mathscr{N}(\mathbcal{M}) =   \mathscr{D}(\mathbcal{M}) - \mathscr{R}(\mathbcal{M}).
\label{eq:nullity}
\end{eqnarray}  
The SSS of Eq.~(\ref{linear_map}) can be expressed by $\mathscr{N}(\mathbcal{M})$  independent solutions  as 
  $\rho_S = \sum^{\mathscr{N}(\mathbcal{M})}_{i=1} p_i \rho_i$, where  $\{\rho_i\}$ span the SSS subspace with arbitrary real coefficients $\{p_i \} \in \mathbb{R}$. They correspond to the zero eigenvalue of the matrix $\mathbcal{M}$, which arises from the multiplicity $\mathscr{N}(\mathbcal{M})$ of identical rows and columns in the matrix. SSS subspaces remain invariant under repeated $\mathbcal{M}$ applications,  which tend to polarize an initially thermal network state. Hence, the SSS dimensionality $\mathscr{N}(\mathbcal{M})$ determines the network polarizability value (Fig.~\ref{fig:2}). The task is then to find the rank $\mathscr{R}(\mathbcal{M})$ which discloses the nullity $\mathscr{N}(\mathbcal{M})$ according to Eq.~(\ref{eq:nullity}).

  Our first general observation, proved by Theorem 1 (Methods 1) is: \\
  
  Theorem 1: Let $\rho^S \equiv \{\rho_1, \rho_2, ..., \rho_N\}$ be a set of steady states generated by map $\mathbcal{M} $:
\begin{equation}\label{A:e1}
 \mathbcal{M}  \rho_i = \rho_i \mbox{~~for~~} i  = 1,2,...,N  
\end{equation}
Then the system will converge to the particular steady state with the same symmetry as the initial state. The different steady states $\rho_i$ depend on the basins of attraction of the map, but may also be degenerate.\\    

Hence, given a symmetric initial state, such as the fully mixed thermal state, a symmetric map only yields a symmetric steady state, and not any of the asymmetric steady states in the SSS subspace, even if they are all degenerate.
  
{\color{black} Although much less demanding than $\hat H_S$ diagonalization, finding the rank $\mathscr{R}(\mathbcal{M})$ (in Eq.~\eqref{eq:nullity}) of the $(4^N -1)$ dimensional matrix representation of $\mathbcal{M}$ requires iterative convergence to its row-echelon form. This procedure may fail even for $N \ge 3$ due to inherent inaccuracy caused by the complex matrix representation of $\mathbcal{M}$. This difficulty motivates the alternative graph-theory approach introduced below.}

 \begin{figure*}
\begin{center}
 \includegraphics[width=0.84\textwidth]{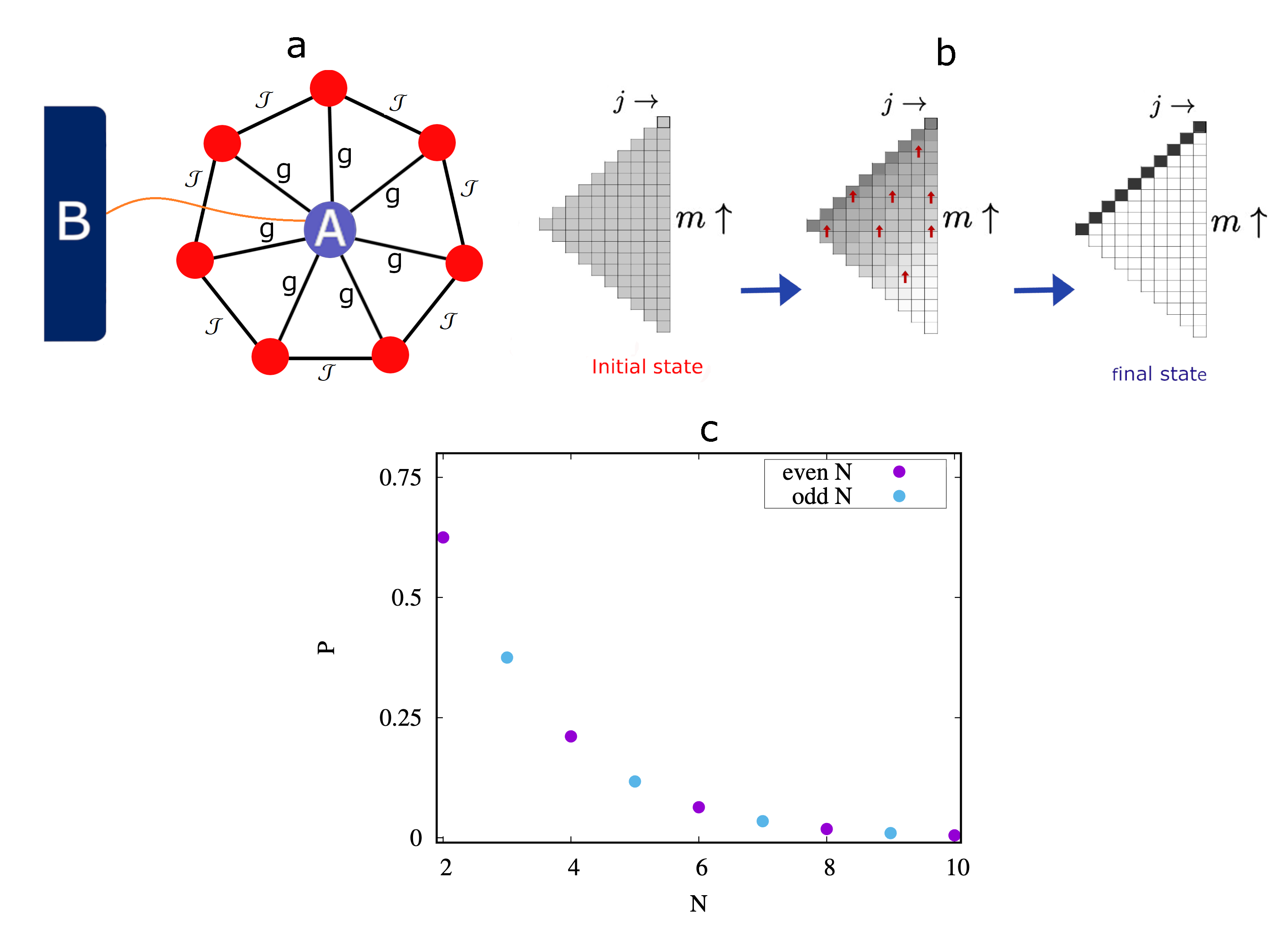}
\end{center}
 \caption{{ \textbf{Angular-momentum constraints on purification:} Angular-momentum symmetry constraints on purification in the star or the equivalent isotropic Heisenberg-chain model: (a) A Scheme that depicts a system $S$ of either isolated or identically-coupled spins that are purified via the ancilla spin $A$ which is repeatedly reset by the cold bath $B$. (b) Block purification scheme of $S$ in angular momentum $j-m$ basis. In the initial fully mixed state, all $j-m$ states in $S$ have equal probabilities, thus all blocks have the same color. The goal is to concentrate the probabilities at the FGS state via $S-A$ resonant transfer (RT). The RT excitation exchange proceeds vertically, from $m$ to $m+1$ in each $j$ block, and is thus confined to conserved-$j$ subspaces (creating symmetry bottlenecks) for purification. (c) Polarizability as a function of the number of spins $N$ under angular-momentum symmetry constraints according to Eq.~\eqref{eqnnP}. }}
      \label{fig:DRT11}
  \end{figure*}

\begin{figure*}
 \begin{center}
\includegraphics[width=0.90\textwidth]{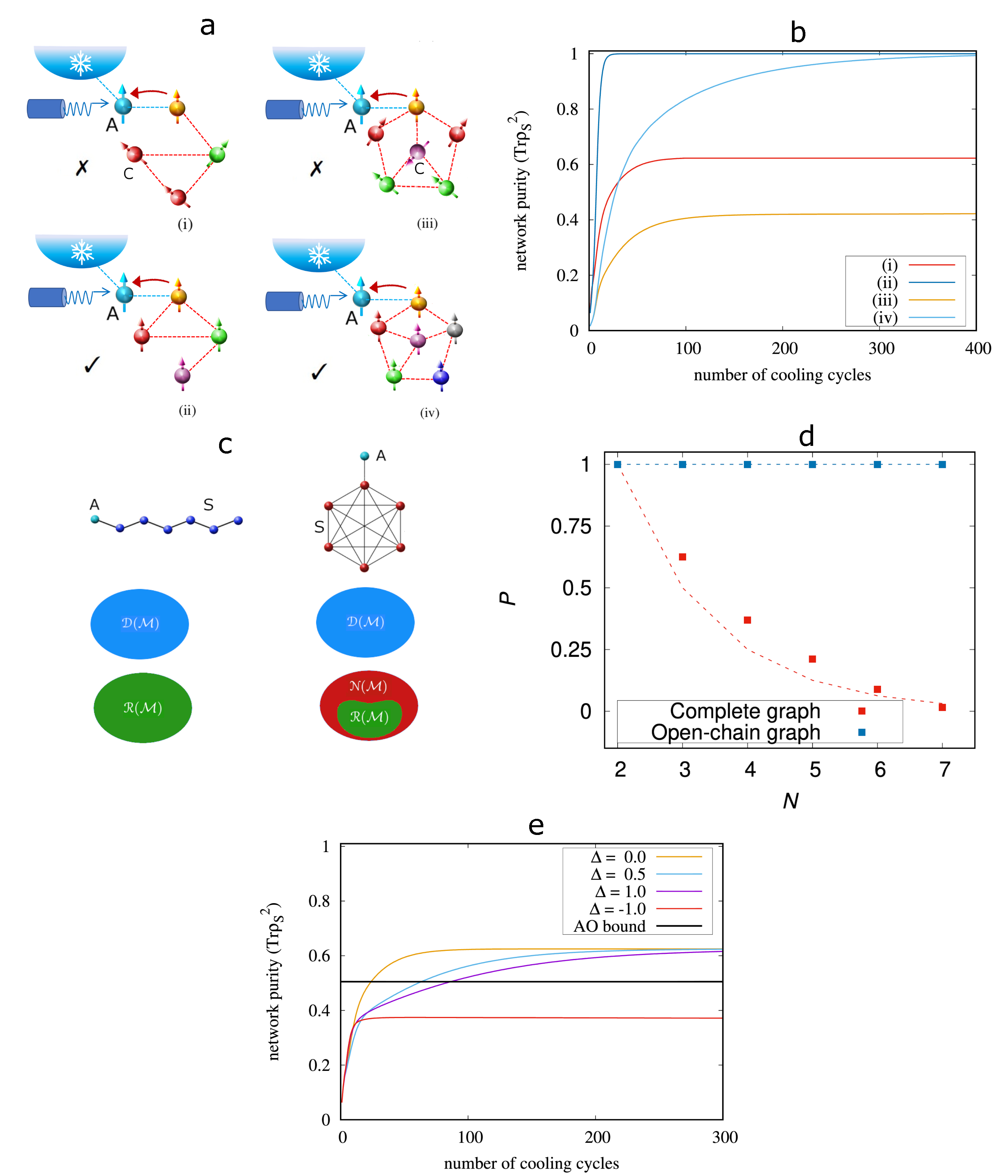}\\
 \end{center}
 \caption{{\textbf{Automorphism constraints on purification:} Polarizable and unpolarizable networks under graph automorphism constraints: (a) Polarizability (\ding{52}) or non-polarizability (\ding{55}) of some representative networks (i)-(iv) via collective entropy swapping with a probe (ancilla) spin $A$ that is intermittently coupled to a cold bath. Network polarizability is obtained by graph-theoretic considerations regarding their automorphism orbits (AO).  Nodes that belong to the same AO, are colored with the same color in the graph, whereas different colors divide the nodes into topologically equivalent sets. Visual inspection of network (i)-(iv) suffices to determine their polarizability bounds. (b) Numerically calculated network purity $\text{Tr~}\rho^2_S$ as a function of the number of ancilla-resets for the networks (i)-(iv) shown in (a).  The calculations confirm our prediction that full purification (polarization) is only achievable for networks with non-degenerate automorphism orbits (AO). (c) left- A polarizable network $N$ coupled to an ancilla $A$ represented by an identity (open-chain) graph for which the rank is equal to the dimensionality  ($\mathscr{R}(\mathbcal{M}) = \mathscr{D}(\mathbcal{M})$), right- an unpolarizable network for which $\mathscr{R}(\mathbcal{M}) < \mathscr{D}(\mathbcal{M})$. (d) Estimated purity versus spin number $N$ for open-chain graphs and complete graphs. Complete graphs (red dotted line) have maximal $\mathscr{N}(\mathbcal{M})$ (Eq.~(\ref{eq:nullity})) and hence the lowest polarizability. (e) Same as (b) for network (i) with different anisotropy $\Delta$ parameters.}} 
 \label{fig:2}
 \end{figure*}

\subsection{Symmetries that prohibit polarization}

As shown in Sec. IIB, in order to fully purify a state of an $N$-spin system ($S$), we must avoid steady-state multiplicity/degeneracy. A central insight is that the condition for steady-state multiplicity corresponding to $\mathscr{R}(\mathbcal{M})<\mathscr{D}(\mathbcal{M})$ is the symmetry of $\mathbcal{M}$ as it acts on $\rho_S$ in Eq.~\eqref{A:e1}. This symmetry may have different origins: it may pertain to the angular momentum of collective spin variables (collective magnetization)~\cite{njp_2020}, to their spatial or topological order~\cite{rudolph_pra_2004}, or to the spectral symmetry of the eigenvalues~\cite{mieghem_2010}, as discussed below. 

\textcolor{black}{It is an open question whether the different symmetries are related to each other and what role they play when two or more symmetries are present simultaneously in the system. In what follows, we focus on each symmetry separately, so as to understand their role in constraining the polarization of a spin network.}


Angular momentum symmetry: The role of angular momentum symmetry can be understood by considering the case of isolated spins or the isotropic Heisenberg  Hamiltonian of $N$ spins. The isotropic Heisenberg model corresponds to the choice $\Delta = 1$ in Eq.~\eqref{network_hamiltonian}. This isotropic Heisenberg model reduces to the isolated-spin (star) model, if the coupling strength between the $S$ spins is negligible compared to that in the $S-A$ interaction, i.e., $\mathcal{J} << \vert g_k \vert$. 

To demonstrate the angular-momentum symmetry effect, we consider a one-dimensional spin chain coupled to an ancilla with equal couplings (Fig.~\ref{fig:DRT11}a). This implies all the couplings are zero except for  $\mathcal{J}_{i,i+1} = 1$ in $\hat{H}_S$ and $g_k = g$ for all $k$ in $\hat{H}_{SA}$. In these cases, the joint dressed eigenstates of the spin-$1/2$ $A$ and the multispin $S$ networks are naturally represented in the basis of their combined angular momentum $\vec{J} = \vec{s} +\vec{S}$, $\vec{s}$ and $\vec{S}$ being the respective angular-momentum operators of $A$ and $S$. These eigenstates are labeled by the total angular-momentum value $j$ and the $z$-axis magnetization $m$, as well as by the degeneracy index $\alpha$. The dynamics governed by $\hat H_{SA}$ in Eq.~\eqref{eq: probe-spin-hamiltonian} (that cause resonant $S-A$ transfer (RT)) is able to break the collective magnetization symmetry of $S$, i.e., $[\hat{H}_S, \hat{S}_z] \ne 0$, but the total $S+A$ magnetization associated with $\vec{J} = \vec{s} +\vec{S}$ is conserved, i.e., $[\hat{J}^2, \hat{H}_{SA}] = 0$. Hence, the requirement of SSS, $[\mathbcal{M},\hat{\Pi}_i] \ne 0$, is not satisfied, where $\hat{\Pi}_i$ are all possible symmetry operations.    

Consequently, the time-evolution operator corresponding to RT cannot mix states with different $j$.
At the $n~$th cycle, the populations of the eigenstates  $p^{(n)}_{\alpha, j,m}$ satisfy the recurrence relations (see SI IV)
 \begin{eqnarray}
   &&p^{(n+1)}_{\alpha,j,m} = p^{(n)}_{\alpha,j,m} \cos^2(J^+_{j,m} \tau) +  p^{(n)}_{\alpha,j,m-1} \sin^2(J^+_{j,m -1} \tau), \nonumber\\
   && J^+_{j,m} =  \sqrt{(j - m) ( j + m + 1)}.
   \label{eq: recurrence}
 \end{eqnarray}
These recurrence relations show that a flow of probabilities $p^{(n)}_{\alpha, j,m}$ takes place within each $j$ sector from lower to higher values of $m$ (Fig.~\ref{fig:DRT11}b). Hence, the blocks with $(j,m=j)$ remain invariant under RT and cause a bottleneck effect that impedes further purification.
Each block with particular $(j,m)$ values retains a certain fraction of the probability $p^{(n)}_{\alpha,j,m}$  and transfers the rest to the block corresponding to $(m+1)$ with the same value of $j$. The block with the lowest $m$  value gains nothing, while the block with the highest value of $m = j$ accumulates all probabilities coming from the lower $m$ blocks (Fig.~\ref{fig:DRT11}b). 
As $n \rightarrow \infty$,  the probabilities accumulate at the highest $m = j$ blocks of each $j$ value. Importantly, this consequence of quantum correlations reflects non-equilibration in the chain during the closed-system evolution stage, contrary to the ETH~\cite{deutsch_eth_2018}. 
Therefore, even after $n\rightarrow \infty$ cycles, the entropy is still high (see SI V), so that RT completely fails to achieve purification.

This bottleneck effect can also be predicted from Theorem~1, (Methods 1), since the system and the ancilla combined have a rotational symmetry. Therefore, the spins in the asymptotic limit $N>>1$ have identical density matrices, so that the steady state of the system can be described only by $(N/2 + 1)$ independent real numbers for even $N$ and $(N+1)/2$ for odd $N$. 

In this case, the steady-state value (asymptotic in time) of the polarization can be given by the exact form (see SI V),
\begin{eqnarray}\label{eqnnP}
 P &=&  \frac{(2N+1)}{4^N} \binom{N}{\frac{N}{2}}, {~~~}\text{even~}N, \nonumber\\
  &=& \frac{(N+3)}{4^N} \binom{N+1}{\frac{N-1}{2}}, {~~~}\text{odd~}N. 
\end{eqnarray}
In the thermodynamic limit $N>>1$, the steady state polarization in Eq.~\eqref{eqnnP} scales as $P \approx \sqrt{N}\exp(-N/2)$ with the system size (Fig.~\ref{fig:DRT11}c).



\begin{figure*}
 \begin{center}
 \includegraphics[width=0.63\textwidth]{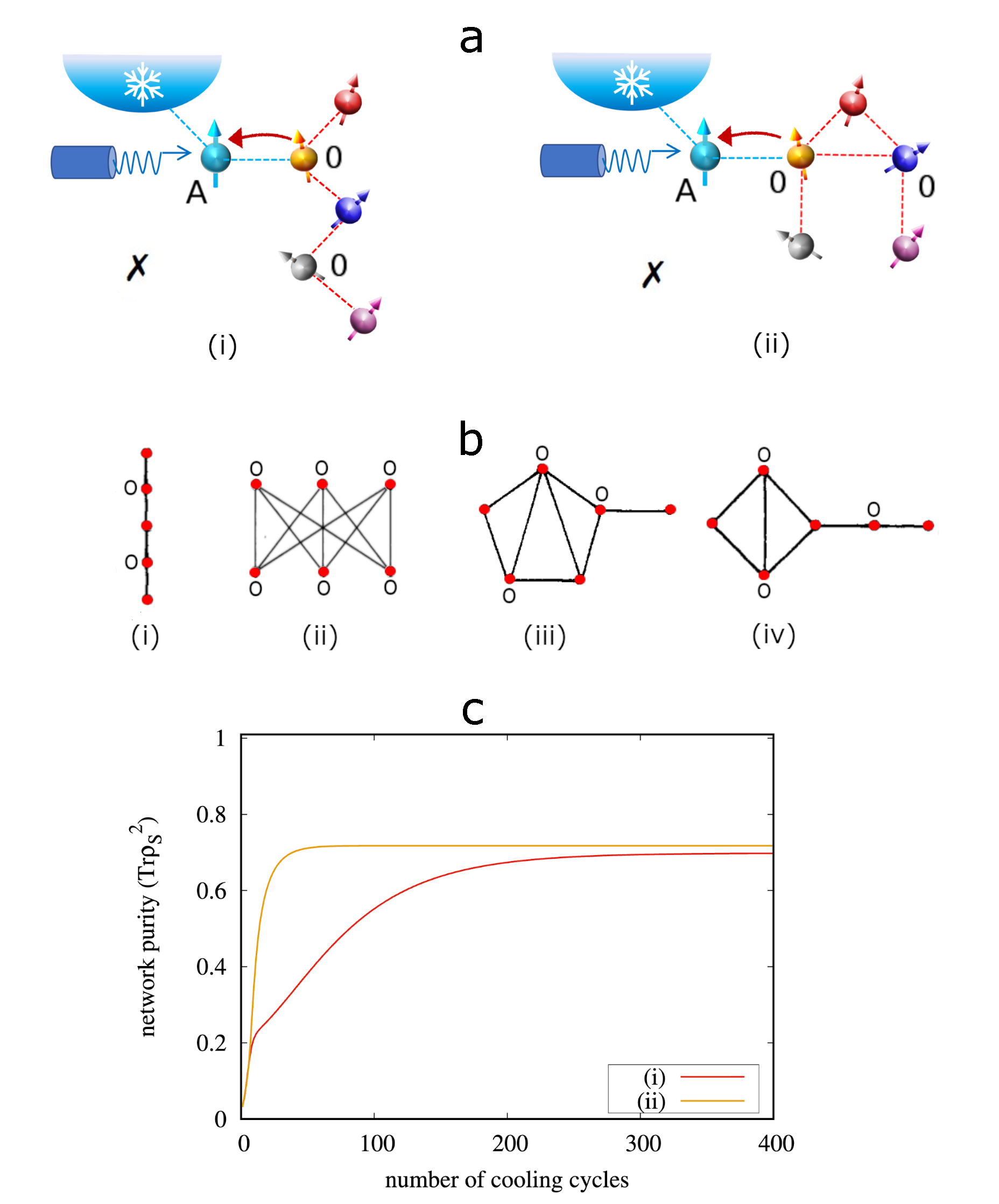}
 \end{center}
 \caption{{\textbf{Spectral symmetry constraints on purification:} Networks where spectral symmetry (SPS) hinders full polarization. (a)  The nodes where the support of the eigenvector(s) corresponding to the null subspace is zero, are marked with $0$. (b) SPS effects in diverse graph classes: (i) Path graph $P_5$ ($N=5$), (ii) complete bipartite graph $K_{3,3}$ ($N=6$), (iii) an identity graph ($N=5$), (iv) a graph with a non-identity nontrivial AO (mirror symmetry) ($N=6$). Nodes marked with `O' denote null support of the kernel. (c)  Numerically calculated time dependence of network purity $\text{Tr~}\rho^2_S$ as a function of the number of ancilla-resets for the networks.  The calculations confirm that full purification (polarization) is only achievable for spin networks with non-degenerate automorphism orbits (AO) that also lack SPS. }}
\label{fig:Delta}
 \end{figure*}


Graph automorphism symmetry: A general class of symmetries plays an important role even in networks without regular geometries. These networks can have any dimensionality, be periodic or aperiodic, and have any connectivity among the spins.

Graph theory allows the determination of the network polarizability value without the need of finding the rank of $\mathbcal{M}$. This is accomplished by mapping the problem to that of finding the discrete symmetries of the network graph representation, since SSS subspaces are ``dark subspaces" imposed by network symmetries:  e.g., in non-interacting (star-model) spin-networks having a collective (Dicke) basis or in  Heisenberg chains~\cite{kurizki_kofman,villazon_2020,bose, nick,alvarez_2015} the initial states are mixtures of either symmetric (bright) or anti-symmetric (dark) state subspaces under permutation. The general question we raise is: how does the network symmetry affect its steady state?

The symmetry of the map in Eq.~({\ref{linear_map}}) is determined by that of $\hat{U}(\tau)$, which can be written as a power series  expansion of $\hat{\phi}(\tau) = \int^\tau_0 \hat{H}(t)dt$, and thus has the symmetry of Eq.~(\ref{network_hamiltonian}) determined by the $J_{ij}$ matrix. 
This observation leads to  a map $\mathbcal{M}$ which generates a graph and its SSS have the same symmetry as the $N \times N$ adjacency matrix $\mathbf{A}$~\cite{Diestel_springer_2006}
\begin{eqnarray}
&&\mathbf{A}(i,j) = 0, {~~~~~~} i=j\nonumber\\
&&\mathbf{A}(i,j) = J_{ij}, {~~~~~~} i \ne j, {~~}\mbox{} {~~}(i,j) \in E\nonumber\\
\label{eq:adjacency}
\end{eqnarray}

A key notion rooted in graph symmetry is the automorphism orbit (AO): a set of equivalent nodes such that the graph remains invariant under their permutation~\cite{grohe_2020,balasubramanian_1980, balasubramanian_1995}. Since map $\mathbcal{M}$ retains the symmetry of $\textbf{A}$, nodes $(i,j)$ that belong to the same AO are interchangeable under $i \leftrightarrow j$ permutations. Thus, AOs and SSS have the same multiplicity, reflecting their common permutation symmetry.  Networks where at least one AO has more than one node are not polarizable for every initial state  (see Theorem~2, Methods 2):   Therefore, a necessary condition for a graph with $N$ nodes to be fully polarized, is for it to have $N$ AOs.

Theorem~2  holds for all canonical spin-network models (the Ising, Heisenberg, XY, SYK, Hubbard models, and more) that possess discrete graph symmetries. In the map (\ref{eq:adjacency}), the ancilla $A$ must break its $SU(2)$ symmetry at the resetting stage, but the graph symmetry of $A+S$ remains intact.

To determine the graph polarizability $P \equiv \text{Tr}(\rho^2_S)$, we consider the initial states in the computational (energy) basis $|s_1,s_2,...,s_N\rangle$ where $|s_i\rangle = \vert 0\rangle_i, \vert 1 \rangle_i$. Out of  $2^N$ computational-basis states, there is a  subset that  satisfies  
\begin{equation}
{\hat{\pi}}_{ij} |s_1,..,s_i,..,s_j,..,s_N\rangle = |s_1,..,s_i,..,s_j,..,s_N\rangle,    
\end{equation}
where $\hat{\pi}_{ij}$ is the exchange operator that preserves the  AO symmetries. This subset consists of $2^K$ states, $K$ being the number of AOs. Such a subset is however only a fraction of an initial thermal mixture. The polarizability $P$ for the maximally mixed initial state is determined by the AO number (degeneracy) $K$ as (see App. 2)
\begin{equation}
P \approx \frac{1}{2^{N-K}}.   \label{eq: bound} 
\end{equation}

This estimated polarizability based on AO symmetry is one of our major results.
{\color{black} It implies that the fractional number of graphs that have $K=N$ increases exponentially with the dimensionality $N$~\cite{erdos_1963}, and the probability to polarize an arbitrary connected graph grows similarly. This fundamental result goes far beyond restating the existence of symmetry subspaces. It provides an explicit, closed-form connection between graph topology and dissipative dynamics, unlike previous formulations of symmetry-restricted cooling or algorithmic reset protocols.}

Among all possible $N$-spin network graphs, a complete graph~\cite{Diestel_springer_2006} is the one with the largest SSS. This graph has only two AOs: one that consists of a spin at node $1$ that is coupled to the probe, and another containing the remaining  $(N-1)$ spins.  Hence, a complete graph has the lowest polarizability 

\begin{equation}
P \approx 1/2^{N-2}.  
\end{equation}
It is realized by a fully connected network with identical couplings (Fig.~\ref{fig:2}c). 

By contrast, an identity  graph that has $N$ different nondegenerate AOs can attain the highest polarizability, 
\begin{equation}\label{7ceqn}
P = 1. 
\end{equation}
This graph is realized, e.g., by an open chain of spins (Fig.~\ref{fig:2}a). The polarizability for diverse graphs is plotted versus $N$ in Fig.~\ref{fig:2}(b), \ref{fig:2}(d).   Importantly, graph polarizability is necessarily determined by its information content which expresses the network diversity (see Theorem~3, Methods 2).

To illustrate these results, we present the number of AOs by distinct colors in diverse networks in Fig.~\ref{fig:Delta}a with $3 \le N \le 6$.  The dependence of the network polarizability on AO degeneracy (Methods, SI.~II, and III) is consistent with the numerics.

If anisotropy (field-bias) expressed by parameter $\Delta \ne 0$ is present in Eq.~\eqref{network_hamiltonian}, the polarizability becomes intractable, as models with $\Delta$ on graphs are unsolvable. The reason is that $\Delta$ acts as self-energy on the diagonal elements of the adjacency matrix $\textbf{A}$, whereas AO symmetry is related to its off-diagonal elements, whose importance diminishes as $|\Delta|$ grows, becoming marginal for $|\Delta| \ge 1$. Thus, AO symmetry can no longer predict the polarizability constraints for $ \Delta = -1$. In the language of graph theory, anisotropic networks are represented by graphs with self-loops. For such graphs, any prediction of its polarizability requires exact diagonalization of the adjacency matrix. Nevertheless as seen from Fig.~\ref{fig:2}(e), graphs with $0 < |\Delta| < 1$, tend to conform with the AO polarizability $P$ in Eqs.~\eqref{eq: bound}-\eqref{7ceqn}.

{\color{black} 
A local transverse field along $z$, which is an on-site potential, may correspond to self-loops with weights given by the field strength~\cite{Mulken2011_CTWQ_review,Babai2016}. 
The symmetry classification by AOs can be generalized to self-loop augmented matrices 
via the theory of vertex-automorphism groups of weighted graphs~\cite{Babai2016, Grohe2018} (Methods 4).  
}


Spectral symmetry: Among all identity graphs that encompass the network and the ancilla, there can exist graphs with other symmetries that prohibit full purification of the network starting from an arbitrary initial state, under resetting operations of the ancilla that do not break the graph symmetry. Namely, the steady states of such graphs have a dark-state (null-support) component that hinders full polarization. This means that apart from the angular-momentum (Sec. IIC1) or the graph automorphism (AO) symmetry (Sec. IIC2), a system can have more symmetries that generate dark states.  These dark states are imposed by spectral symmetry (SPS)~\cite{mieghem_2010}.  SPS in a graph is conditioned on the existence of at least one zero eigenvalue of the adjacency matrix $\mathbf{A}$,which can be confirmed by the singularity of $\mathbf{A}$, i.e., $ \mbox{Det~} \mathbf{A} = 0$. In such a graph, said to be singular~\cite{sciriha2007c}, the nodes at sites $x$ that have null support correspond to the null solutions of $N$ linear homogeneous algebraic equations, $\mathbf{A}  \textbf{x} =  \textbf{0}$.  According to Theorem~4 (Methods 3) and Theorem~5 (Methods 3), in graphs possessing SPS, the ancilla does not break the symmetry of the null subspace of the original graph and hence is unable to fully polarize the network.

Graph theory and linear algebra are unable to connect  SPS and AO symmetries as they cannot answer: a) What are the conditions for a graph to have a zero eigenvalue in the spectrum? b) If the graph has a zero eigenvalue, does this eigenvalue have zero support? c) If there is zero support, at which node is it?  The spectrum of an arbitrary graph can be found by diagonalizing the adjacency matrix~\cite{mieghem_2010}, and there are no general alternative recipes for answering questions (a-c).

Still, for certain graph classes the existence of a zero eigenvalue can be predicted from the global spatial symmetry:  (i) The eigenvalues of a path graph with $n$ vertices are given by $E_j = 2 \cos \frac{\pi j}{n+1}, j = 1,...,n$. Hence, it has one zero eigenvalue for odd $n$. This class pertains to the singly-excited collective Dicke states of a 1d lattice (open periodic chain) of two-level systems, which include a dark (subradiant) state and $(n-1)$ bright (superradiant) states~\cite{rudolph_pra_2004}. (ii)  A  complete bipartite graph $K_{m,n}$, composed of two sets of nodes of sizes $n$ and $m$, respectively,   has a spectrum with  eigenvalue zero of multiplicity $(n -2)$, and two non-zero eigenvalues $E_{\pm} = \pm\sqrt{mn}$. 

As seen from Fig.~\ref{fig:Delta}, we can have an identity or a non-identity graph with a null subspace; implying that SPS and AO symmetries are indeed generally unrelated.
From numerical searches we deduce that identity graphs with SPS occur only for $N \ge 5$ (Fig.~\ref{fig:Delta}) and are rare under realistic spin-spin interactions:  their occurrence probability is of measure zero for $N>>1$, hence their practical significance is marginal.

Nevertheless, it is important that SPS may prohibit the polarization of a system of isolated spins connected to an ancilla via arbitrary (even random) couplings by resonant transfer (RT), as further discussed in Sec. IID. 


 \begin{figure*}
\begin{center}%
  \label{}
  \includegraphics[width=0.69\textwidth]{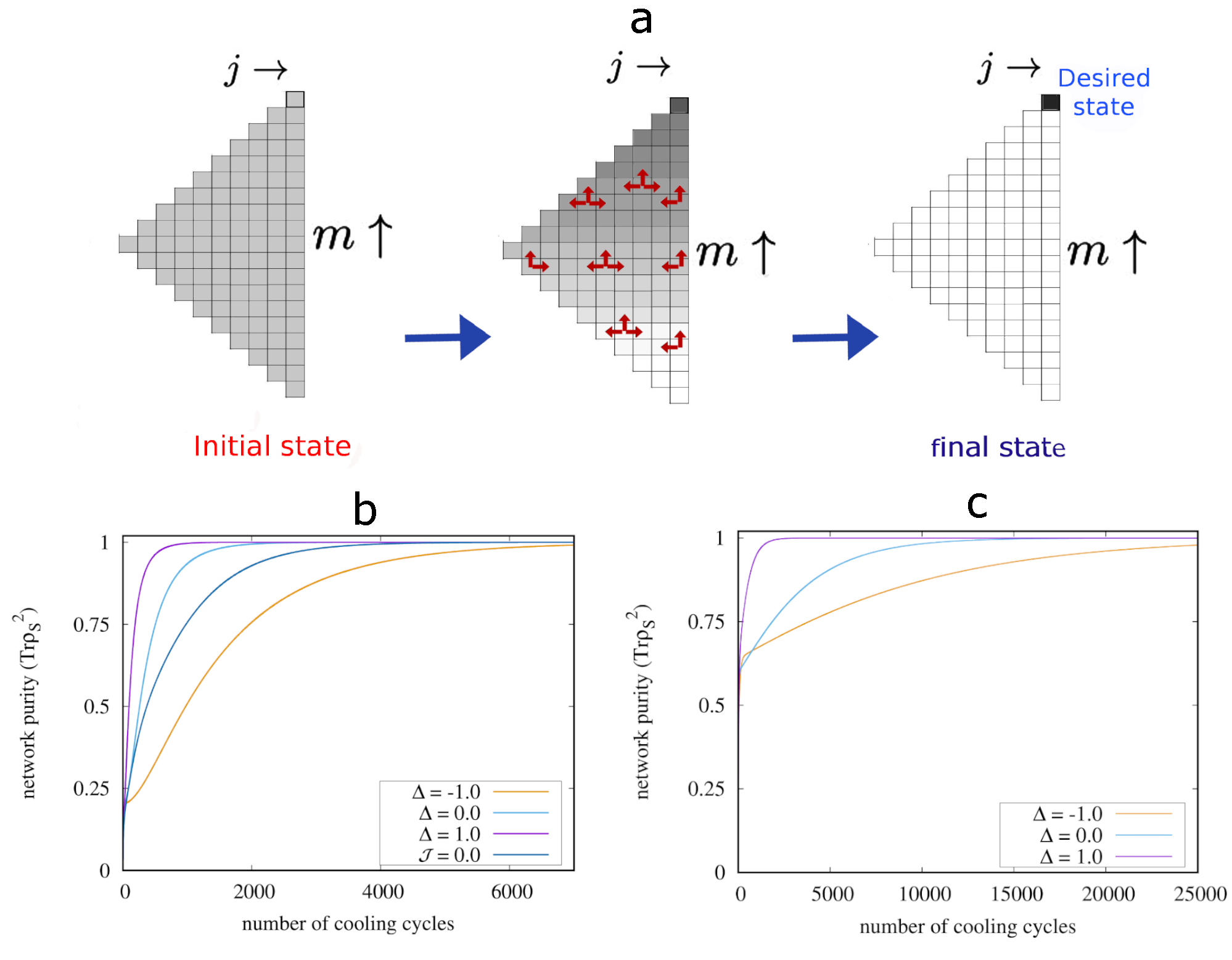}
\end{center}
 \caption{{\textbf{Purification using ADRT protocol:} (a) Schematic representation of the ADRT purification protocol for a star model: a system $S$ of isolated spins via the ancilla spin $A$, showing its overwhelming ability to overcome symmetry constraints/bottlenecks compared to RT in Fig.~\ref{fig:DRT11}. In the ADRT protocol, the excitation exchange takes place both horizontally and vertically (i.e., along $m$ and $j$), thus mixing all $j$-blocks. This allows us to achieve the desired final state. (b) The variation of the network purity with the number of cycles for the isolated spin model and the Heisenberg chain of $5$ spins with different anisotropy parameters $\Delta$. (c) The variation of the network purity with the number of cycles for the non-polarizable graph (i) shown in Fig.~\ref{fig:2}(a) with different anisotropy parameters $\Delta$. Both plots (b) and (c) show that the desired state is attained using the ADRT protocol, unlike the RT protocol used in Fig.~\ref{fig:2}(b). }}
      \label{fig:ADRT1}
  \end{figure*}

\subsection{Purification beyond symmetry constraints}

We may conclude from Sec. IIC that any symmetry leads to degeneracies~\cite{vandersypen_2005} in the matrix elements of $\rho^{(n\rightarrow \infty)}_S$ which are obtained by solving the set of linear homogeneous equations (SI V). The resulting matrix $\rho^{(n\rightarrow \infty)}_S$ is then composed of invariant blocks that are impervious to purification by $\mathbcal{M}$, if it commutes with $\hat{H}_S$. Can one overcome this severe constraint on collective purification?

Our proposed universal solution to this problem enables us to retain a unique steady state, i.e., to have $\mathscr{R}(\mathbcal{M})=\mathscr{D} (\mathbcal{M})$, by choosing $\hat{H}_{SA} (t)$ that breaks the symmetry of $\hat{H}_S$.
To this end, $\hat{H}_{SA}$ must be restricted by the conditions
\begin{eqnarray}
 [\mathbcal{M},\hat{\Pi}_i] \ne 0 ~~\forall i \mbox{~~iff~~} [\hat{H}_{SA},\hat{\Pi}_i] = 0. 
 \label{condition_2}
\end{eqnarray}
where $\hat{\Pi}_i$ are all possible symmetry operations under which the $\hat{H}_S$ remains invariant. 

The purification strategy must lift the symmetry of magnetization conservation by $\hat{H}_S$ along one or more orthogonal axes, so as to mix the subspaces (blocks) of $\rho_S$ that are invariant as long as
\begin{equation}
    [\mathbcal{M}, \hat H_S]=0
\end{equation}

Generic many-body spin Hamiltonians that conserve the total magnetization along one of the spin axes, say in $z$ direction, include the broad class of anisotropic Heisenberg chains, the long-ranged non-nearest-neighbor resonant dipole-dipole interaction (RDDI)~\cite{corzo2019waveguide,Gershon,defenu2023long,opatrny2023nonlinear,thiru}, the spin-fermion Hubbard model~\cite{PhysRevLett.3.77,PhysRevLett.65.2462} and more.

We now propose a purification strategy that can lift all symmetry constraints via three consecutive steps in each cycle:

Step 1: Resonant $S-A$ transfer: An interaction Hamiltonian that causes resonant transfer (RT) between $S$ and $A$ is activated from the beginning to the middle of the cycle in the form (see Eq.~\eqref{eq: probe-spin-hamiltonian})~\cite{hartmann_hahn_1962} 
\begin{eqnarray}
\label{eqn10g}
 \hat{H}^{res}_{SA} (t)= \frac{1}{2} \sum^N_{k=1} g_k(\hat{\sigma}^+_A \hat{\sigma}^-_k +\hat{\sigma}^-_A \hat{\sigma}^+_k), \,\, n\tau\leq t\leq n\tau +\tau/2.\nonumber\\
 \label{drt_a}
\end{eqnarray}
This RT maximizes the excitation exchange (swap) and guides $S$ toward the FGS state, but it generally fails to break all the symmetries in the $\hat{H}_S$-Hamiltonian. It thereby precludes full purification because each symmetry gives rise to invariant blocks of $\rho_S$, limiting the polarizability of $S$. To surpass it, we have to resort to the next step of the protocol. 

Step 2: Dispersive $S-A$ coupling : We abruptly change the $\hat H_{SA}$ say, halfway through the cycle: the joint basis (of $S$ and $A$) is suddenly rotated by $90^\circ$ with respect to the hypersphere axes, so as to render the exchange interaction off-resonant (dispersive) in the form

\begin{eqnarray}
 \hat{H}_{SA}(t) = \hat{H}^{disp}_{SA} =  {\sigma}^{z}_A \sum^N_{k=1} \tilde{g}_k \hat{\sigma}^{z}_k,~~~ n\tau + \frac{\tau}{2} < t < (n+1) \tau. \nonumber\\   
 \label{eq: H_SA_b}
\end{eqnarray}
\label{eq: H_SA}
Explicitly, the cycle employs a pair of non-commuting, consecutive interaction Hamiltonians which we dub alternate dispersive and resonant transfer (ADRT). It thereby breaks the combined $S+A$ system symmetries, in order to lift the steady-state degeneracy.

{\color{black}  Even though the maximally mixed state that we choose to initialize the protocol is uncorrelated, the subsequent application of the non-local interaction Hamiltonians $H^{res}_{SA}(t)$ and $H^{disp}_{SA}(t)$ dynamically generate inter-spin and spin–ancilla correlations. These correlations are unavoidable by virtue of the system-ancilla coupling,
which entangles the ancilla with collective excitation modes of the network. Because of symmetries present in the system, these correlations are constrained to form superpositions within symmetry-protected subspaces, resulting in correlated steady states that are not fully polarizable~\cite{Zanardi1997,Lidar1998,Albert2016, Baumgartner2008}. 

The purpose of the ADRT protocol is to break these symmetry-induced correlations by alternating resonant and dispersive operations, which generate a non-commuting Lie algebra of operators that couple otherwise disconnected subspaces. This procedure converts correlated steady states into effectively decohered populations, enabling full relaxation to the desired state.
}
 
We now present a pivotal proof that the $\hat{H}_{SA}$ in Eqs.~\eqref{drt_a} with equal $g_k = g$ and~\eqref{eq: H_SA_b} with unequal $\tilde{g}_k$ are the only choices that ensure complete purification to the FGS of any $S$ subject to any model that has the FGS state as the eigenstate.\\

{\color{black}
 \setcounter{customthm}{5}
 \begin{customthm}
 \label{theorem:6}
Consider bilinear ancilla--system coupling Hamiltonians that  
preserve the excitation, then the coupling Hamiltonians~\eqref{drt_a}  with equal $g_k = g$ and~\eqref{eq: H_SA_b} with unequal $\tilde{g}_k$ provide one experimentally realizable non-commuting pair that satisfies
\begin{equation}
(\hat{H}_{SA} + \hat{H}_S)\vert0\rangle_A\vert00\ldots0\rangle_S
= E\vert0\rangle_A\vert00\ldots0\rangle_S,
\label{eq: b1}
\end{equation}
and the symmetry-breaking condition Eq.~\eqref{condition_2}
where the desired state $\vert00\ldots0\rangle_S$ is assumed to be an eigenstate of the system Hamiltonian $\hat H_S$.
 \end{customthm}
}
\vspace{2ex}

{\color{black} \noindent \textit{Proof.} 
We restrict our attention to bilinear ancilla--system interactions that preserve the total excitation number.
For each bond $(A,k)$, the most general excitation-preserving bilinear Hamiltonian is then confined to the span
$ 
\{\sigma_A^+\sigma_k^- + \sigma_A^-\sigma_k^+,\;\sigma_A^z\sigma_k^z\}.$
Accordingly, the most general excitation-preserving bilinear ancilla--system coupling takes the form
$$ 
\hat{H}_{SA}\in \mathrm{span}\{\hat{H}^{res}_{SA}(\{\phi_k\}),\hat{H}^{disp}_{SA}\}, $$
where
\begin{align}
\hat{H}^{res}_{SA}(\{\phi_k\})
&= \sum_{k} g_k
\big(
e^{i\phi_k/2}\hat{\sigma}_A^+\hat{\sigma}_k^-
+ e^{-i\phi_k/2}\hat{\sigma}_A^-\hat{\sigma}_k^+
\big),
\label{span_xy}
\\
\hat{H}^{disp}_{SA}
&= \sum_{k} \tilde{g}_k\,\hat{\sigma}_A^z\hat{\sigma}_k^z .
\label{span_disp_sa}
\end{align}


The central (ancilla) spin acts as a source of static gradient fields, producing the Ising-type interaction $\hat H^{disp}_{SA}$. 
When driven by a transverse microwave field, it induces the resonant flip-flop interaction $\hat H^{res}_{SA}$ in the rotating frame. Either resonant (transverse) or off-resonant/dispersive (longitudinal) coupling are experimentally accessible at any given time, i.e. \eqref{span_xy} and \eqref{span_disp_sa} are mutually exclusive and linear combinations of these two interaction types are therefore excluded. 
  {\color{black} Importantly, we note that the commutator between $\hat H^{res}_{SA}$ and $\hat H^{disp}_{SA}$ does not vanish in general for $k>1$, i.e., $[\hat H^{res}_{SA}(\{\phi_k\}),\hat H^{disp}_{SA}]\neq 0$ generically. 
}

The desired state $\vert 0\rangle_A\vert00\cdots0\rangle_S$ lies in the eigenspectrum of both Hamiltonians, so that  
\begin{equation}\label{eqn23aaa}
\hat{H}^{res}_{SA}\vert0\rangle_A\vert00\cdots0\rangle_S = 0,
\end{equation}
\begin{equation}\label{eqn24aaa}
\hat{H}^{disp}_{SA}\vert0\rangle_A\vert00\cdots0\rangle_S
= \Big(\sum_k \tilde g_k\Big)\vert0\rangle_A\vert00\cdots0\rangle_S .
\end{equation}

In the Hartmann--Hahn (HH) regime of the resonant Hamiltonian~\eqref{drt_a} 
polarization transfer between the ancilla and the $k$th spin occurs when their precession frequencies are matched in the rotating frame. 
Consequently, under the standard HH assumptions of uniform couplings and a common phase (for instance, as in Dzyaloshinskii--Moriya (DM) type antisymmetric exchange, where the two flip--flop contributions acquire opposite relative signs), the transverse interaction can be taken (up to a collective $z$-rotation) in the experimentally natural form~\eqref{eq: 12a}.
In this scenario, all spins experience identical resonance  provided: 
(a) uniform couplings $g_k=g$;
(b) equal phases $\phi_k=\phi$; and
(c) removal of the global phase $\phi$ by a collective rotation about the $z$ axis,
resulting in the flip--flop Hamiltonian
\begin{equation}
\hat{H}^{res}_{SA}
= g\sum_{k=1}^N
\big(
\hat{\sigma}_A^+\hat{\sigma}_k^-+
\hat{\sigma}_A^-\hat{\sigma}_k^+
\big).
\label{eq: 12a}
\end{equation}

The Hamiltonian~\eqref{eq: 12a} is invariant under spin-exchange operation 
among the spins. 
Let $\hat{\pi}_{ij}$ denote the corresponding spin-exchange operator defined as $\pi_{ij} \vec{\sigma}_i = \vec{\sigma}_j \pi_{ij}$.
Then $ [\hat{H}^{res}_{SA},\hat{\pi}_{ij}]=0$  expressing spin-exchange symmetry of the bath spins. If the longitudinal couplings $\tilde{g}_k$ are also uniform,
then $\hat{H}^{disp}_{SA}$ is likewise spin-exchange symmetric,
$[\hat{H}^{disp}_{SA},\hat{\pi}_{ij}]=0$.

Both these Hamiltonians also commute with the total spin operator of the ensemble, i.e., $[\hat{H}^{res}_{SA},\hat{J}^2]=0$ and $ [\hat{H}^{disp}_{SA},\hat{J}^2]=0$, where
\begin{equation}
\hat{J}^2
=
\hat{J}_x^2+\hat{J}_y^2+\hat{J}_z^2,
\qquad
\hat{J}_\alpha=\frac12\sum_{k=1}^N\hat{\sigma}_k^\alpha ,
\end{equation}
so that the full dynamics remains block-diagonal in the Dicke basis labeled by total spin $k$.
Thus, uniform transverse and longitudinal couplings preserve both angular-momentum (Dicke) symmetry and the permutation symmetry of the spin ensemble.

To break the spin-exchange and angular-momentum (Dicke) symmetries while maintaining HH resonant transfer, the longitudinal couplings $\tilde g_k$ must be nonuniform. For this choice, we have  $[\hat{H}^{disp}_{SA},\hat{J}^2] = i \sigma^z_A \sum_{k<l}(\tilde{g}_k - \tilde{g}_l)(\sigma^y_k \sigma^x_l - \sigma^x_k \sigma^y_l) \ne 0$ and $[\hat{H}^{disp}_{SA},\hat{\pi}_{ij}] = i(\tilde{g}_i - \tilde{g}_j) (\sigma^y_i \sigma^x_j - \sigma^x_i \sigma^y_j) \ne 0$. This choice confirms that the full evolution  satisfies to the symmetry breaking condition~\eqref{condition_2}.



Since $\vert00\cdots0\rangle_S$ is an eigenstate of $\hat H_S$ with eigenvalue $E_S$, we have from \eqref{eqn23aaa}, \eqref{eqn24aaa}:
\begin{align}
(\hat{H}^{res}_{SA}+\hat H_S)\vert0\rangle_A\vert00\cdots0\rangle_S
&= E_S\vert0\rangle_A\vert00\cdots0\rangle_S,
\\
(\hat{H}^{disp}_{SA}+\hat H_S)\vert0\rangle_A\vert00\cdots0\rangle_S
&= \Big(E_S+\sum_k\tilde g_k\Big)
\vert0\rangle_A\vert00\cdots0\rangle_S .
\label{eq: b2}
\end{align}

Hence, under excitation-preservation, HH condition, and experimental-mutual exclusion of \eqref{span_xy} and \eqref{span_disp_sa} assumptions, the Hamiltonians~\eqref{drt_a} with equal $g_k = g$ and~\eqref{eq: H_SA_b} with unequal $\tilde{g}_k$ form the unique experimentally realizable non-commuting pair satisfying condition~\eqref{eq: b1}.}

{\color{black} Although bilinear system--ancilla couplings that do not conserve excitation number
(e.g., \ terms proportional to $\hat{\sigma}_A^-\hat{\sigma}_k^-$ or $\hat{\sigma}_A^+\hat{\sigma}_k^+$)
may destroy the product state $\ket{0}_A\ket{00\cdots0}_S$, such terms cannot support purification. The operators $\hat{\sigma}_A^-\hat{\sigma}_k^-$ remove population only from already doubly excited states and produce no directed population flow toward the target state. Consequently, bilinear couplings that do not preserve excitation number are counterproductive for our purification protocols.}\\






Step 3: Ancilla cooling under dynamical modulation: The dissipative resetting of $A$ dumps part of the entropy acquired from $S$ into an ultracold (photonic or phononic) bath, thus breaking the $SU(2)$ symmetry of $S$  during this step, which completes the cycle (Fig.~\ref{fig:circuit}). 
At this step, the ancilla qubit $A$ is connected to the ultracold bath where it can dump the entropy acquired from $S$.
The state of the ancilla following the dissipative step in the n-th cycle has the form

 \begin{eqnarray}
 \rho^{(n)}_{A} = \begin{pmatrix} 
 1-z^{(n)} & x^{(n)} \\
x^{*(n)} & z^{(n)}
\end{pmatrix},
\label{eq: rdm_probe}   
\end{eqnarray}   

The evolution of $z^{(n)}$ and $x^{(n)}$ is determined by the (generally non-Markovian) relaxation and the decoherence rates $\gamma(t)$ and $\Gamma(t)$, respectively~\cite{xia_2024}, and by a dynamically induced phase $\int dt' \Omega(t')$ due to a driving (control) field with Rabi frequency $\Omega(t)$. The resulting cooling/purification speed of $A$ can be maximized upon choosing the modulation of $\hat{H}_A(t)$ to be such that the rates $\gamma(t)$, $\Gamma(t)$ conform to the anti-Zeno regime of enhanced interaction with the bath~\cite{kurizki_kofman,kofman_2000}.

\begin{figure}
\begin{center}
\includegraphics[width=0.43\textwidth]{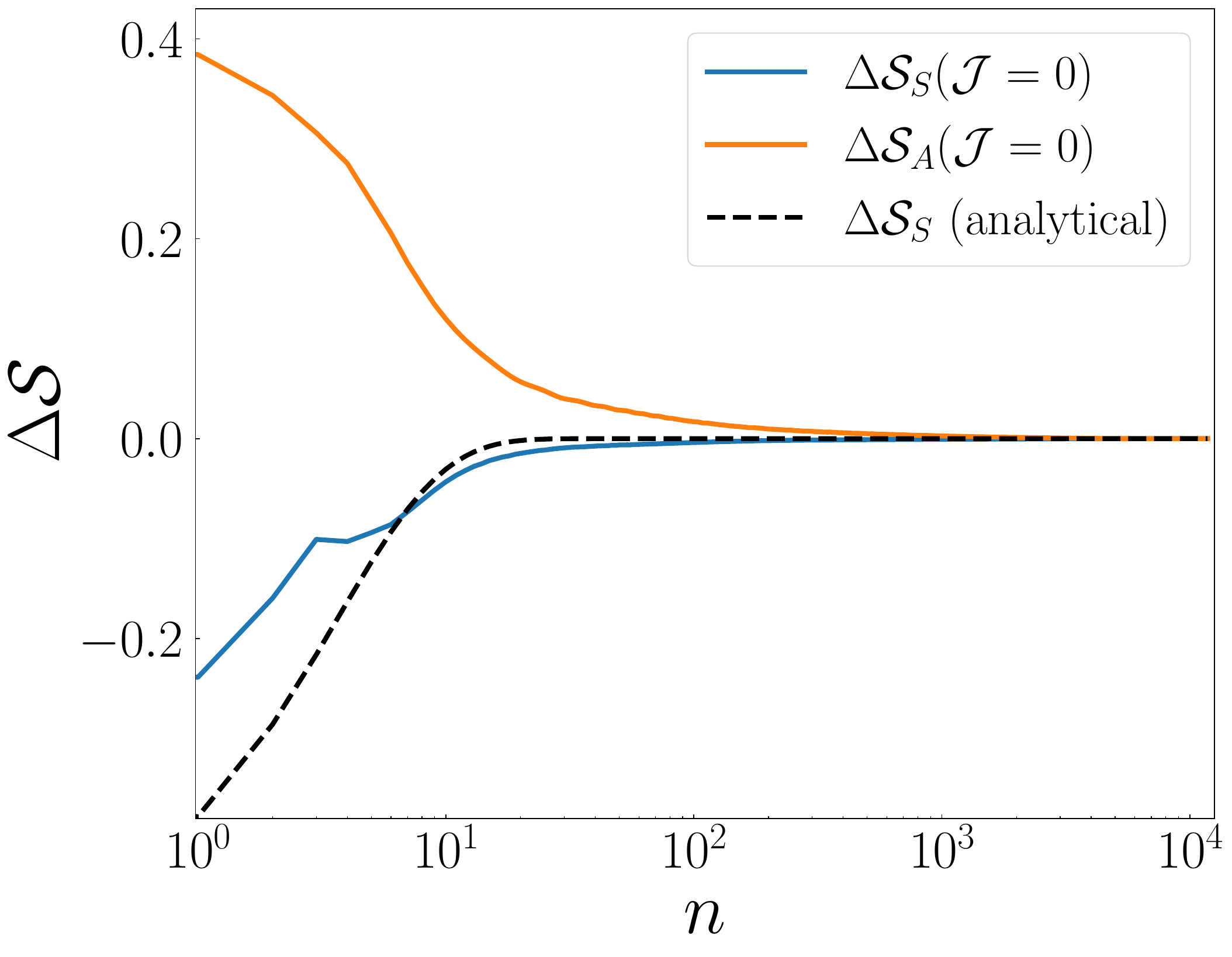}
  \\
\end{center}
 \caption{{\textbf{Purification speed and the third law:} Change in purity of $S$ and $A$ states with the number of cycles $n$ for $N=6$ isolated spins ($\mathcal{J} = 0$), which is the same as for the isotropic model ($\Delta = 1$) chain under ADRT, shows that both purities saturate for $n \gtrsim 10^2$, cycles, so that the ancilla purity can probe the arrival of the system at the FGS.  The dashed black line denotes the estimated analytical curve with a proportionality constant of $-0.5$, signifying its power-law approach towards the FGS.}}
      \label{fig:DRT111}
  \end{figure}

\subsection{Purification speed and the third law}


To demonstrate the convergence of the dynamics to the FGS, we estimate the asymptotic behavior of the purification/cooling rate for the isolated-spin (star) model or the isotropic Heisenberg chain ($\Delta = 1$).  The ADRT dephasing randomizes the phases corresponding to the same $m$ values. The probabilities of these states are then equalized by the dephasing for sufficiently large-$N$ multispin limit, provided the couplings $g_k$ are different and incommensurate.

Initially, the populations of any 
$j$-block $p^{(n+1)}_{j, m} \sim \cos^{2}(J^+_{j,m} \tau) p^{(n)}_{j, m}$ (see Eq.~\eqref{eq: recurrence}) decrease with the cycle index $n$, except for that of the desired FGS. By contrast, in the asymptotic limit of large $n$, let $\varepsilon(n)$ be the error in the fidelity of the desired FGS after $n$ cycles. Here we assume that the off-diagonal populations are negligible compared to their diagonal counterparts, $N$ is sufficiently large such that the number of blocks with $m=N/2-1$ is $d_{j=N/2-1}>1$, and the populations of the blocks with $m< N/2-1$ are negligible compared to the blocks with $m =  N/2-1$. The asymptotic rate of change in the population in the FGS is then 
\begin{eqnarray} 
p^{(n+1)}_{N/2, N/2} -  p^{(n)}_{N/2, N/2}
 \approx  \varepsilon(n) \left(  \frac{  \cos^2(J^+_{N/2, N/2-1} \tau ) +  d_{j=N/2-1}}{ d_{j=N/2-1} + 1 }     \right)^n. \nonumber\\
\end{eqnarray}
 The rate of change in the entropy of the system  is thus proportional to
\begin{eqnarray}
\Delta \mathcal{S}_S  =  \mathcal{S}^{(n+1)}_{S}- \mathcal{S}^{(n)}_{S} 
  \propto \mathcal{O} \Bigg(\left(  \frac{  \cos^2(J^+_{N/2, N/2-1} \tau) +  d_{j=N/2-1}}{ d_{j=N/2-1} + 1 }     \right)^n \Bigg).\nonumber\\
\end{eqnarray}    



The expression in the brackets is constant for given $N$ and $\tau$. Hence, the rate of entropy change of the system decreases with the number of cycles $n$ as a power-law and so does the purity change (Fig.~\ref{fig:DRT111}). This implies that an infinite number of cycles, $n \rightarrow \infty$ is required to reach the FGS with ideal ($100\%$) fidelity, consistently with the third law of thermodynamics.

Thus, the third law of thermodynamics~\cite{nernst_2}, is here confirmed by the exact dynamics of a multi-spin system that obeys either the star model or the isotropic Heisenberg model.  This result stands in contrast to previously considered approximated solutions for a magnon chain 
 coupled to a qubit that is cooled by a bosonic bath~\cite{kolar_2012}. However, the general question of the asymptotic cooling speed in multi-spin networks remains open, as their dynamics is commonly unsolvable.

\subsection{Experimental Protocols}
The general cooling (purification) protocol described here is experimentally realizable for interacting electron-spin networks composed of adjacent quantum dots or NV centers in diamond that interact with nuclear spins. {\color{black} Polarization transfer from optically pumped electron spins to nuclear spins, by SWAP gates (double-resonance experiments) or level-anticrossing schemes, is the prevalent technique for nuclear spin initialization~\cite{emsley_book,bernstein_1957,vandersypen_2005,dasari2022anti}. However, both of these methods respect the symmetries of the problem. Hence, for resetting an entire nuclear spin-network to the state $\vert 000 ... \rangle$, the presence of symmetries in the electron-nuclear spin coupling hinders complete polarization so that the state $\vert 000 ... \rangle$ cannot be approached. Thus, one needs an additional process to break the symmetry-induced bottleneck for complete polarization, which in our ADRT protocol is provided by the dispersive (dephasing) step. 
Such dephasing can be achieved by adding a period of free evolution between in two consecutive (SWAP) operations. Namely, our procedure involves the repeated applications of a resonant field (RT) for polarization transfer (Step 1) and subsequent free evolution of the interacting spins (Step 2) 
that causes symmetry breaking due to their different coupling constants. This procedure is currently feasible via refocusing techniques that can achieve selective, tunable couplings between spins~\cite{dasari2022anti}. }

A particularly simple example is the star model realizable for NV-center electron spin that acts as a central-spin ancilla to cool the surrounding $^{13}\text{C}$ nuclear spins. 
{\color{black} The central spin acts both as a source of gradient fields (conforming to Eq. \eqref{eq: H_SA_b}) through its dipolar couplings to the surrounding nuclei and as a hyperpolarizing agent (conforming to Eq. \eqref{drt_a}) when driven by a microwave field perpendicular to the gradient axis. Thus, switching on the transverse microwave field implements the XY-type interaction of Eq. \eqref{drt_a}, which enables polarization transfer in the rotating frame, whereas switching it off restores the gradient-like (Ising-type) interaction of Eq. \eqref{eq: H_SA_b}. Crucially, these two Hamiltonians cannot coexist, since one is converted into the other via the application of the external driving field.

An analogous scenario can be realized in Rydberg-atom networks: by controlling the detuning of the exciting lasers, interactions between designated Rydberg atoms can be selectively activated or deactivated \cite{steinert_2023}.
} 

An excellent candidate for spin-network engineering consists of molecules, each having two electron spins separated by a few nm, known as molecular rulers~\cite{azarkh2019gd}. Due to the high stiffness of these molecules, their electron spins are well localized and resilient to phonon-induced decoherence, even at appreciable temperatures~\cite{reginsson2012trityl}. Their controllable concentration in a solution to be deposited on a diamond surface can create interacting spin networks on demand, as follows:
A magnetic tip that generates fields with nm-scale gradients~\cite{schein2025pulsed} can be used to control the spin-network order, connectivity, and geometry by switching on and off their intramolecular and intermolecular coupling. These network spins can be cooled/purified from their initial thermal state via their dipolar coupling to the closest electron spin of a shallow NV center close to the diamond surface ($<$ 10 nm) that acts as an ancilla spin. The higher gyromagnetic ratio of the electron spins compared to nuclear spins would allow to drive them selectively at high frequencies (via ODMR) and detect their response (Fig.~\ref{fig:exp}a).

Remarkably, the mere visual inspection of a molecule allows us to conclude whether it is polarizable due to lack of symmetry as in glucose, or unpolarizable due to its complete symmetry, as in benzene (Fig.~\ref{fig:exp}b).

\begin{figure*}
\begin{center}
\includegraphics[width=0.77\textwidth]{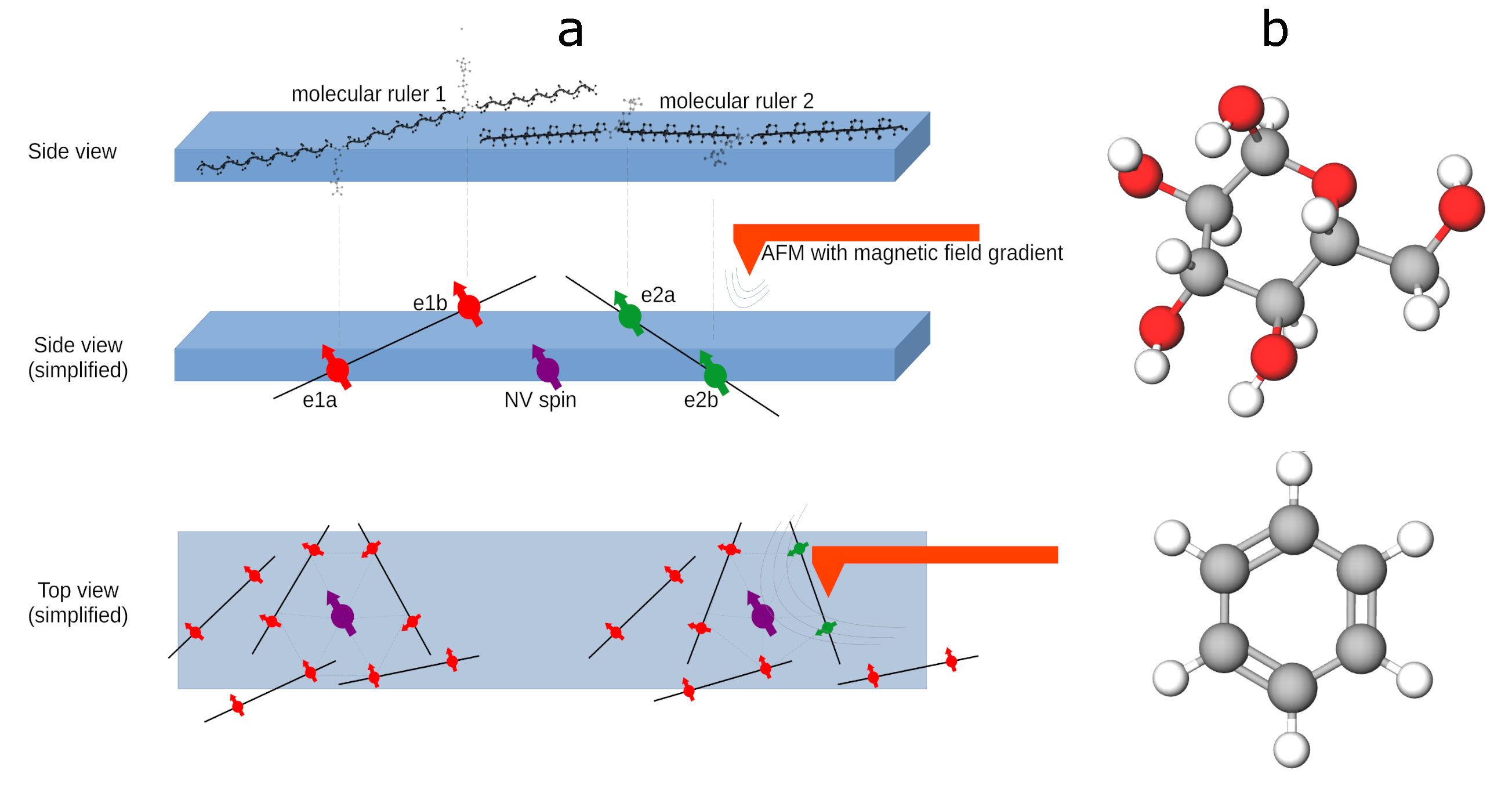}
\end{center}
 \caption{{ \textbf{Experimental demonstrations:} (a) A schematic of an envisaged experimental implementation of an addressable spin network using an NV center in diamond coupled to molecular rulers. Top panel: Shows the surface of the diamond with a couple of molecular rulers~\cite{reginsson2012trityl}, each with a pair of spin labels (electron spins). Middle panel: Portrays an abstraction of the top panel and introduces a magnetic field gradient from an AFM tip~\cite{schein2025pulsed}, which can be turned on and off with a sub-microsecond speed. The gradient allows for a selective addressing of electron spins from the host of molecular rulers on the surface of the diamond by tuning the Larmor precession frequency of those electrons (here e1a and e1b;   e2a and e2b) and bringing them into resonance with the pulse sequence's resonance condition. Bottom panel: A simplified top-view of the diamond surface, where a proper choice of the AFM's tip position and current (gradient) enables specific network topologies, e.g., a hexagonal (left) or a pentagonal (right) network. (b) Top--- Glucose molecule: an example of a spin network that is fully polarizable, due to lack of symmetry among the spin nodes. Bottom--- Benzene molecule: an example of a spin network that is not polarizable due to its symmetry.}}
      \label{fig:exp}
  \end{figure*}
 
\section{Discussion}


The conceptual and technological importance of cooling down interacting quantum many-body systems, particularly spin chains~\cite{endres2016atom,loss_1998,petta2005coherent,kandel2021adiabatic,qiao2020coherent,majer2007coupling,yamamoto2017coherent,johnson2011quantum,king2022coherent,bose,subra1,christandl,wang,manmana,chiara,kurizki_2015_pnas,meier_prl_2003,benjamin2003quantum,fel1997multiple,jones2011quantum,bonizzoni2024quantum,Pritam2024QST,yang2022variational,kolar_2012, kurizki_kofman,miguel_2018,Steinlechner2013,laflamme_2022,boykin_pnas_2002,nernst_2,reviewPr,njp_2020}, from a thermal to a pure state provides ample motivation for tackling this challenging theoretical problem. 
Such purification is the prerequisite for quantum annealing~\cite{johnson2011quantum,king2022coherent},  quantum memory~\cite{loss_1998,bose,christandl,wang,manmana,chiara}, and scalable quantum computing or simulation~\cite{meier_prl_2003,benjamin2003quantum, georgescu_2014} realizable by e.g., cold atoms~\cite{endres2016atom}, quantum dots~\cite{loss_1998,petta2005coherent,kandel2021adiabatic,qiao2020coherent}, superconducting qubits~\cite{majer2007coupling} or their optical analogs~\cite{yamamoto2017coherent}.  The same requirements arise for molecules whose structure and functionality are unraveled by nuclear magnetic resonance (NMR)~\cite{fel1997multiple,jones2011quantum} as well as for spin chains used in quantum sensing~\cite{schlipf_2017,bonizzoni2024quantum,Pritam2024QST}, metrology~\cite{yang2022variational}, and thermometry~\cite{kolar_2012, kurizki_kofman,Chattopadhyay_2025}.

{\color{black} We have here addressed a basic difficulty in designing spin-network resetting/purification, namely, that not only are the spectra and the eigenstates of interacting spins are in general insolvable, but even the numerical analysis of their open-system (quantum) time evolution is computationally unfeasible, say for chains with more than few spins.  
To cope with this difficulty, our analysis has established fundamental relations between the graph structure of quantum spin networks and their polarizability. 

Specifically, we have shown that 1) only graphs that allow for a \textit{unique steady state} are fully polarizable irrespective of their initial states, 2) the degeneracy of automorphism orbits (AO) in the graph sets bounds on network purification and thereby reduces the insurmountable exponential time complexity of diagonalizing the $2^N$ dimensional Hamiltonian of N-node networks to a quasi-polynomial number of steps $\mathcal{O}(N^{p \log N})$ for some polynomial $p$~\cite{grohe_2020}. \textit{Visual inspection of a graph suffices to reveal its AO degeneracy and thus its polarizability} for networks with $N \le 4$. \textit{Rare exceptions} of measure zero for $N>>1$ are graphs with \textit{global spectral symmetry} (SPS)~\cite{mieghem_2010}, that possess dark states which cannot be purified. SPS is generally unrelated to AO symmetry that involves $(i,j)$ exchange operations in the network.

Beyond algorithmic efficiency, the AO formalism provides a conceptual gain that the rank-based approach cannot offer: it establishes a \textit{direct connection between network topology and its maximal polarization}, 
and thus provides a physical interpretation of why certain geometries resist full polarization. 
}

To exceed the symmetry-imposed polarizability bound, we have formulated the universal guiding principle of purifying an interacting $N$-spin system by overcoming symmetry constraints even if its spectrum and dynamics are untractable. This principle amounts to \textit{ breaking all symmetries of the system}, a feature that has \textit{no classical analog}. To this end, we have formulated the alternate dispersive-resonant transfer (ADRT) protocol that destroys such symmetries and thereby \textit{restores the eigenstate thermalization hypothesis} (ETH)~\cite{deutsch_eth_2018,rigol2008thermalization,cramer2008exact} by allowing to fully purify spins in the network to the desired pure state, efficiently and rapidly. It consists of \textit{alternating resonant and off-resonant (dispersive) coupling} of the ancilla $A$ to the system (network) $S$, which induce, in turn, resonant flip-flop and off-resonant dephasing interactions in $S$. As long as the $S-A$ couplings $g_k$ form a \textit{non-uniform} set,  \textit{the spin-network symmetry can be completely broken} by ADRT, and $S$ polarization can be attained. 

 {\color{black} At an abstract level, ADRT shares the essential structure of collision models (CMs) that have been introduced for thermalizing machines\cite{scarani2002,ciccarello2017,strasberg2019, lorenzo2017, cattaneo2021,PhysRevE.99.042145}, wherein the system undergoes successive interactions with an ancilla and each such interaction is reset by an activated dissipation channel.  While ADRT fits this repeated-interaction paradigm, it differs from standard CMs both conceptually  and practically: 
(a) In most CMs, ancillae locally interact with a single subsystem or a small subset thereof~\cite{scarani2002, ciccarello2017}. In contrast, ADRT employs a \textit{single ancilla} that couples \textit{collectively} to many sites of the system in each cycle. This collective coupling is essential for redistributing entropy among symmetry-related degrees of freedom that are not considered in standard CM.
(b) CMs usually assume a fixed interaction Hamiltonian for each collision or allow its conditioned changes~\cite{ciccarello2017}, but do not exploit an alternating non-commuting sequence of interactions designed to break the system symmetries, which is the goal of ADRT. 
(c) A faithful mapping of the dynamics to a collisional model is far from trivial for an interacting qubit network, since many ancillary qubits and collisions with multiple nodes would then be required.}

Importantly, our approach bypasses the difficulty of measuring the spin-chain purity by inferring it from the measured ancilla purity. The purification speed under the ADRT protocol for the isotropic Heisenberg chain or the star network has been shown to diminish exponentially with the number of cycles, thus confirming the third law. One must still clarify the rapport with conflicting results whereby a magnon chain can be cooled by a driven qubit (ancilla) that is coupled to a bosonic bath at a speed that \textit{remains finite at arbitrarily low temperature}~\cite{kolar_2012}.

The collective $N$-spin purification strategy based on Eqs.~(\ref{drt_a}), (\ref{eq: H_SA_b}) is generally superior in terms of speed and efficacy to the algorithmic cooling approach ~\cite{tal_mor_pra_2016,boykin_pnas_2002, boykin_pnas_2002,oftelie_prxq_2024} which relies on \textit{sequential swaps} between adjacent spins (an approach that has not been attempted for $N$-spin chains). Indeed, each swap in algorithmic cooling has a duration of $\sim \pi/K$, where $K$ is the induced (switchable) coupling strength between adjacent spins. In order to be controlled, the swap rate $K$ should greatly override the natural spin-spin coupling strength $\mathcal{J}$ (cf. Eq.~\eqref{network_hamiltonian}), which may be demanding. Yet even if this condition holds, the shortest duration of sequential swaps, $\sim \pi N/K$ is slower than the collective swap duration $\sim \pi/g$ unless $K \gtrsim Ng$. Furthermore, sequential swaps yield an untractable, non-monotonic entropy distribution in a closed interacting spin chain, which may therefore require a much longer cooling process than the collective ADRT.

To sum up, our general strategy for interacting-network purification rests on three conceptual innovations: i) classification of purification (polarization) constraints due to symmetry by graph theory; ii) an estimator of graph polarizability based on its topology; and iii) a universal dynamical protocol (ADRT) for complete purification by removing the symmetry constraints.
It offers both a comprehensive conceptual framework and a practical pathway toward efficient qubit-network initialization and hyperpolarization beyond the limitations of passive relaxation, free of symmetry obstructions.

At the practical level, ADRT relies solely on repeated, tunable interactions between a single qubit and a network and is therefore adaptable to diverse platforms.  Thus, it can be used to purify and reset a) superconducting qubit networks that are coupled by transmission-line resonators with coupling strength that is tunable via magnetic flux~\cite{liao2010single}; (b) micromaser schemes generalized to multiple coupled cavities that may form purifiable photonic networks; 
c) ordered arrays of two-level systems that are coupled by electromagnetic waveguides (e.g., cold atoms trapped in optical lattices, impurities in crystalline solids or superconducting qubit arrays). All of the above are describable by our general theory of symmetric spin networks that have to be reset after entangling/ QIP operations by repeatedly dumping their entropy in a collective fashion via an optically active ancilla.

{\color{black} 

}

As an outlook, our approach can pave the way to the design of quantum networks capable of efficient initialization (resetting) control for quantum information processing~\cite{kurizki_2015_pnas,georgescu_2014},  quantum-battery designs~\cite{mondal_2022}  or magnetic-resonance spectroscopy~\cite{fel1997multiple,jones2011quantum} and imaging (MRI)~\cite{fel1997multiple,jones2011quantum}. In the context of thermodynamic control~\cite{kurizki_kofman}, the attainable purity and purification speed can determine the efficiency and refrigeration power of a single-spin (ancilla) coupled to a multi-spin bath.

\section{Methods}

\subsection*{1. Map SSS and AO multiplicity of graphs: symmetry constraints}
 For $N$ spins, the $2^N$ dimensional Hilbert space, corresponds to the following number of independent real variables in $\rho_S$ and hence in the map $\mathbcal{M}$ (Eq.~(\ref{linear_map})):
\begin{eqnarray}
\mathscr{D}(\mathbcal{M}) =\sum^N_{m=1} 3^m {N \choose m} = 4^N -1.
\label{eq:a1}
\end{eqnarray}\\





\noindent Theorem 1. Let $\rho^S \equiv \{\rho_1, \rho_2, ..., \rho_N\}$ be a set of steady states generated by map $\mathbcal{M} $:
\begin{equation}\label{A:e1a}
 \mathbcal{M}  \rho_i = \rho_i \mbox{~~for~~} i  = 1,2,...,N  
\end{equation}
Then the system will converge to the particular steady state with the same symmetry as the initial state. The different steady states $\rho_i$ depend on the basins of attraction of the map, but may also be degenerate.\\

\noindent Proof: Let the map $\mathbcal{M} $ be symmetric with respect to some symmetry operations $\Pi$:
\begin{equation}\label{A:e2}
 \Pi \mathbcal{M}  \Pi^\dagger = \mathbcal{M} .   
\end{equation}

Then the state  $(\Pi \rho_i)$  preserves the symmetry of the map
\begin{eqnarray}\label{A:e4}
 \mathbcal{M} (\Pi \rho_i) =  \Pi \rho_i, 
\end{eqnarray}

and belongs to the set $\rho^S$.

(a) If the state $\rho_i$ is invariant under $\Pi$, then $\Pi \rho_i = \rho_i$. which means that there exists only one (unique) steady state. 

(b) If the state $\rho_i$ is not invariant under $\Pi$, then $\Pi \rho_i = \rho_j$ where $j \ne i$.  This would imply that the steady states are related by symmetry operations. 

We assume that the symmetry operation is discrete, i.e., $\Pi^q = \mathbf{1}$ for some natural number $q$. From Eq.~\eqref{A:e4}, it follows that the map can only  have $q$ steady states ($q$-fold degeneracy).\\

Let  the map for one cycle be defined as follows;
\begin{eqnarray}\label{A:e5}
 \mathbcal{F}  \rho^{(n)} = \rho^{(n+1)}, {~~~~} n = 0, 1, 2, .....  
\end{eqnarray}

From the above equation, we have
\begin{eqnarray}\label{A:e6}
 \Pi \mathbcal{F}  \Pi^\dagger \Pi \rho^{(n)} = \Pi \rho^{(n+1)}.    
\end{eqnarray}

Using Eq.~\eqref{A:e2} we can write
\begin{eqnarray}\label{A:e7}
  \mathbcal{F}  (\Pi \rho^{(n)}) = \Pi \rho^{(n+1)}.
\end{eqnarray}

If the initial state is symmetric w.r.t $\Pi$, then
\begin{eqnarray}\label{A:e8}
  \Pi \rho^{(0)} =  \rho^{(0)}. 
\end{eqnarray}

From Eq.~\eqref{A:e7} we have
\begin{eqnarray}\label{A:e9}
 \mathbcal{F}  \rho^{(0)} = \Pi \rho^{(1)}.   
\end{eqnarray}

Yet from the definition Eq.~\eqref{A:e5}, we have $\mathbcal{F}  \rho^{(0)} = \rho^{(1)}$. This means that the state $\rho^{(1)}$ generated by $\mathbcal{F} $ from $\rho^{(0)}$ is also symmetric, i.e., $\Pi \rho^{(1)} = \rho^{(1)}$. 

By induction,  the above statement is valid for all the recurring steps (cycles) $\mathbcal{M} = \mathbcal{F}^{(n \rightarrow \infty)}$ between the initial state and the steady state, hence the steady state is also symmetric. 

This result upholds case (a),  where the steady state is unique if it satisfies $\Pi \rho_i = \rho_i$.

 \subsection*{2. AO degeneracy}

Let us consider the effects of $\hat{\pi}_{ij} =   \frac{1}{2}(\mathbf{1} + \hat{\sigma}^x_i \hat{\sigma}^x_j + \hat{\sigma}^y_i \hat{\sigma}^y_j + \hat{\sigma}^z_i \hat{\sigma}^z_j)$ the Dirac spin-exchange operator that permutes $i$ and $j$. We require that
 \begin{eqnarray}
 [\hat{U}(\tau),\hat{\pi}_{ij}] =  [\hat{H}(t),\hat{\pi}_{ij}] = 0,
 \end{eqnarray} 

The permutations of all the nodes except the node labeled 1 that is connected to the probe (and thus belongs to a separate AO), generate a symmetry group $S_{N-1}$ with the node indices $2,...,N$.  One can construct $(N-1)!$ adjacency matrices $\textbf{A}_1, \textbf{A}_2, ...., \textbf{A}_{(N-1)!}$ from $\textbf{A}$ in Eq.~(\ref{eq:adjacency})
by these permutations. 

Permutations can be represented in the compact cyclic notation that describes the effect of repeatedly applying the permutation $\hat{\pi}$ on the set of node indices $2,...,N$. In this notation,  any permutation among the node indices that belongs to $S_{N-1}$  does not change the graph (represented by its adjacency matrix). 

All symmetry operations of a given graph can be represented by a product of cycles. A cycle of $L$ indices,  referred to as an $L$-cycle, corresponds to an AO of $L$ nodes.  By contrast, if all the indices map to themselves under $\hat{\pi}$, the permutation group has $(N-1)$ one-cycles. In this case all the adjacency matrices are different except for $\textbf{A}_1 = \textbf{A}$. They represent an identity graph that has exactly $N$ AOs with one node each. 


Polarizability: If we could restrict the dynamics to a subspace with $\mathscr{N}(\mathbcal{M}) = 0$, by choosing certain initial states, a fully polarized state would be achieved.  Such a subspace would consist of initial states that respect the AO symmetry, i.e.,
\begin{eqnarray}
  \hat{\pi}_{ij}\rho(0)\hat{\pi}^{\dagger}_{ij} = \rho(0) \forall(i,j)~~[\hat{\pi}_{ij},\hat{H}_N] = 0. 
  \label{eq:commutation}
\end{eqnarray}\\

\noindent Theorem 2. Spins at nodes belonging to the same AO have identical steady states under map $\mathbcal{M}$ provided that the initial state is symmetric, e.g. the maximally mixed state. 

\noindent Proof: Let a pair of nodes $(i,j)$ belong to the same AO. At the steady state, the reduced density matrices of the spins at the nodes $i$ and $j$ are $\rho_{S_{i}}$ and $\rho_{S_{j}}$ respectively. Let us suppose that $\rho_{S_{i}} \ne \rho_{S_{j}}$. Then automorphism permutation $(i,j) \rightarrow (j,i)$ would keep the graph invariant, but the network would have a different steady state where $\rho_{S_{i}} \rightarrow \rho_{S_{j}}, \rho_{S_{j}} \rightarrow \rho_{S_{i}}$. Yet,  since   $\mathbcal{M}$ is a linear map, the steady state must have the same symmetry as the initial symmetric state and thus be unique. Therefore, the supposition is false and $\rho_{S_{i}} = \rho_{S_{j}}$. Thus, reduced density matrices of spins corresponding to nodes in the same AO are identical and hence mutually dependent. This decreases the number of independent variables in $\rho_S$ and consequently reduces the rank $\mathscr{R}(\mathbcal{M})$ of the map $\mathbcal{M}$, thus precluding full polarization.

It can be seen from Eq.~(\ref{eq: map_n}), Theorems~ 1  and 2 that the map $\mathbcal{M}$ retains the symmetry of the AOs as $[e^{-i\hat{H}\tau},\hat{\pi}_{ij}] = 0$ if the pair $(i,j)$  belongs to an AO. The dynamics, if started from a state with AO symmetries, will respect the AO symmetries throughout the evolution. Hence, the dynamics is always restricted to the subspace $\mathscr{R}(\mathbcal{M}) = \mathscr{D}(\mathbcal{M})$, $\mathscr{N}(\mathbcal{M}) = 0$.
  
If there are $K$ AOs in the network, the number of such basis states will be $2^K$. Since  only initial states that preserve the AO symmetry  lead to the fully polarized state, the polarizability (purification) parameter $P$ in the case of the maximally mixed state in the computational basis, is estimated by

\begin{equation}
P = \frac{1}{2^{N-K}},  
\label{eq: M7}
\end{equation}
which is our major result.

Numerical calculations show that asymmetric states make small contributions to polarization and that spectral-symmetry (SPS) states that hinder the polarization are rare. Therefore, Eq.~(\ref{eq: M7}) is the approximated value of Eq.~(\ref{eq: bound}).  This value falls off exponentially with $N$ for a fixed number $K$ of AOs.

Among all possible networks with $N$ spins, the one with the largest steady-state subspace is dubbed a complete graph. In a network with $N$ spins, the maximum and minimum values of $K$ are $N$ and 2, respectively. Hence, the lowest value for $P$ will be achieved for a complete graph, such that $P \approx 1/2^{N-2}$ (Fig.~\ref{fig:2}c).\\

\noindent Theorem 3.  The information contained in a graph $I_G$ and the size of the steady-state subspace with nullity $\mathscr{N}(\mathbcal{M})$ are in one-to-one correspondence.\\

\noindent Proof: The information contained in a graph $I_G$ is a one-to-one map of the distribution of nodes among different  AOs $(n_1,n_2,...,n_N)$. This follows from the concavity property of Shannon entropy. We show in SI C that the nullity $\mathscr{N}(\mathbcal{M})$ is also a one-to-one map of the same distribution.
Since both $I_G$ and the nullity $\mathscr{N}(\mathbcal{M})$ are one-to-one maps of the distribution $(n_1,n_2,...,n_N)$, they are in one-to-one correspondence with each other, i.e., for graphs $G$ and $G'$ we have

\begin{eqnarray}
  \mathscr{N}(\mathbcal{M})_G = \mathscr{N}(\mathbcal{M'})_{G'}   \Leftrightarrow I_G = I_{G'}.
\end{eqnarray}

Therefore, the information in the graph representation of a quantum network uniquely characterizes the SSS properties. Moreover, since $\mathscr{N}(\mathbcal{M})$ is a strictly convex function of $(n_1,n_2,...,n_N)$ (SI B), $\mathscr{N}(\mathbcal{M})$ monotonically decreases with $I_G$. Therefore, the lesser $I_G$, the larger its SSS subspace. This trend is seen from the plot of the relative size of the steady-state subspace, $\mathscr{N}(\mathbcal{M})/(4^N -1)$, versus $\tilde{I}_G = {I_G}/{\log N}$ for all possible networks with $2 \le N \le 7$ spins in Fig.~\ref{entropy_nullity}. 

\begin{figure}
 \begin{center}
\includegraphics[width=0.48\textwidth]{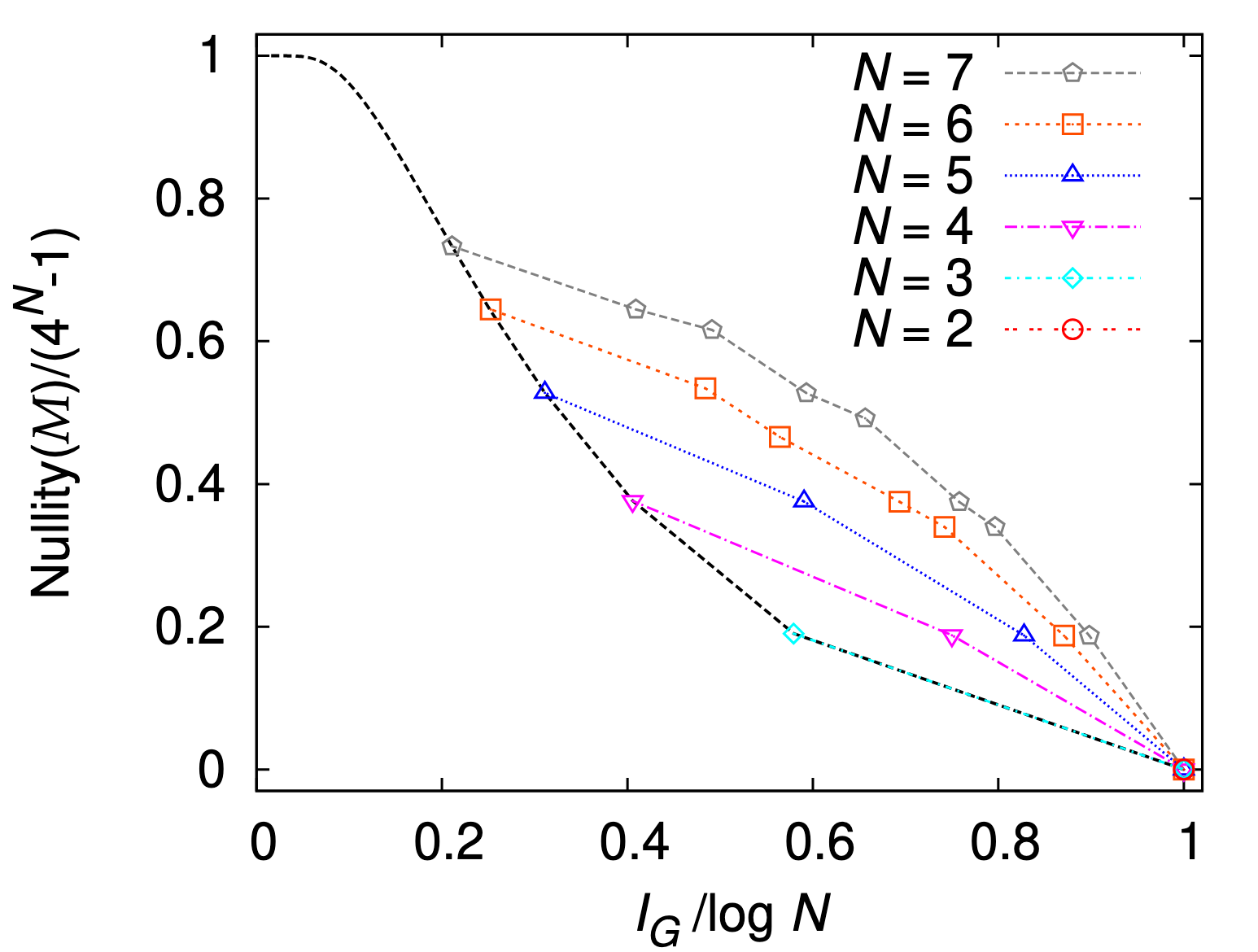}
 \end{center}
 \caption{{ \textbf{Steady state subspace and Information of a graph:} Information contained in a graph $\tilde{I}_G = I_G/(\log N)$ versus the relative size of the steady subspace i.e., $\mathscr{N}(\mathbcal{M})/(4^N-1)$ for all possible $N$-networks with $N=2, 3, 4, 5, 6,$ and $7$. The black dotted line denotes the upper bound of $\mathscr{N}(\mathbcal{M})/(4^N-1)$ given in SI II, Eq.~(S11).}}
 \label{entropy_nullity}
 \end{figure}

\subsection*{3. Spectral symmetry of graphs}

The existence of such a null space depends on the topology of the graph, but it is generally impossible to find the locations of the zeros in these eigenvectors from visual symmetries~\cite{mieghem_2010}.\\

\noindent Theorem 4. If a graph is singular ( that has at least one zero eigenvalue) and the corresponding eigenvectors (kernel of the graph) have null support on any of the nodes (i.e., zeros appear in the eigenvector), the same graph with the probe attached to that node(s) will also contain a null subspace and the eigenvector corresponding to the zero eigenvalues will have null support on the same node.\\

\noindent Proof: Let $\mathbf{x}_0$ be an eigenvector belonging to the eigenvalue zero of the adjacency matrix $\mathbf{A}$ for a graph $G$ with $n$ nodes. If this eigenvector has  a zero in position (site) $(\mathbf{x}_0)_{i} = 0$, then from the eigenvalue equation one has

\begin{eqnarray}
 \sum_{j} \mathbf{A}_{ij} (\mathbf{x}_0)_j = \mathbf{0}, ~~~i = 1,2,...,n
 \label{rank_eq}
\end{eqnarray}

If we add another node $a$  to one of these nodes (say node $1$). The new adjacency matrix $\tilde{\mathbf{A}}$ for a graph $\tilde{G}$ with $(n+1)$ nodes will be

\begin{eqnarray}
&&\tilde{\mathbf{A}}_{ij} =  \mathbf{A}_{ij}, \nonumber\\
&&\tilde{\mathbf{A}}_{a1} = \tilde{\mathbf{A}}_{1a} \equiv g, \nonumber\\
&&\tilde{\mathbf{A}}_{ai} = \tilde{\mathbf{A}}_{ia} = 0,~~~~ i \ne 1. \nonumber\\
\end{eqnarray}

One can readily find from Eq.~(\ref{rank_eq}) that
\begin{eqnarray}
 \text{Det~} \tilde{\mathbf{A}} =   \text{Det~} {\mathbf{A}} = 0.  
\end{eqnarray}
\begin{figure*}[htpb]
 \begin{center}
\includegraphics[width=0.650\textwidth]{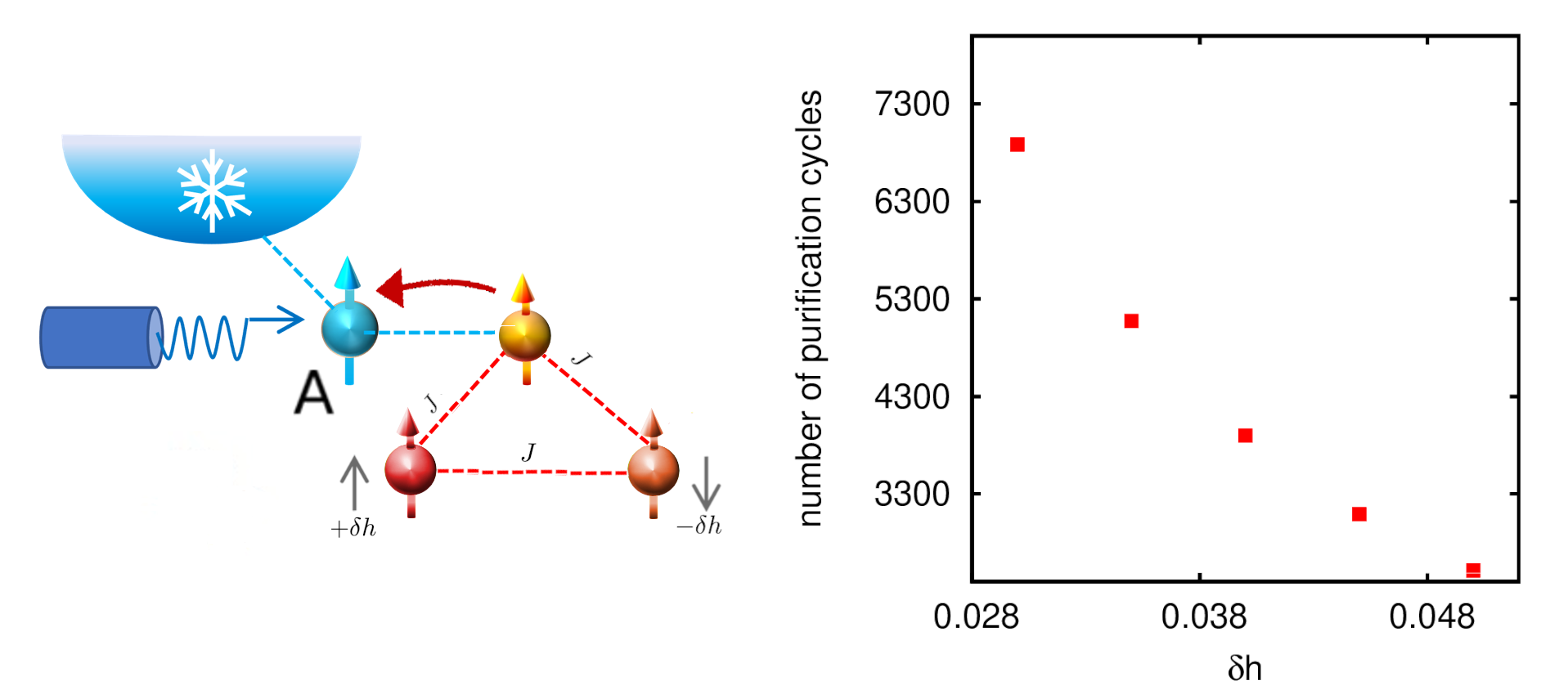}
 \end{center}
 \caption{{\color{black}{\textbf{Symmetry-breaking by on-site field:} Number of purification cycles required to reach $P = 0.9$ in a 3-spin closed chain (isosceles) under slight asymmetry (inset) as a function of  $\delta h$ (transverse-field induced spin-level splitting).  Here, $\delta h$  (diagonal) disorder breaks the AO symmetry and leads to full purification.}}}
\label{fig:9a}
 \end{figure*}

Expanding the row corresponding to the node $a$  of the eigenvalue equation or the graph $\tilde{G}$  associated with the zero eigenvalue 

\begin{eqnarray}
 \sum_{ij}\tilde{\mathbf{A}}_{ij} (\tilde{\mathbf{x}}_0)_j  = \mathbf{0}.
 \label{eq:eigenvalue_eqn_g_tilde}
\end{eqnarray}
 one finds,

 \begin{equation}
(\tilde{\mathbf{x}}_0)_1 = 0. 
 \end{equation}
This proves that the support of the eigenvector with zero eigenvalue of the graph $\tilde{G}$ is null at the node $1$.\\

\noindent Theorem 5. The degenerate subspace of the graph corresponding to a zero eigenvalue contains only the eigenvector corresponding to the zero eigenvalue(s) of the graph without the probe.\\

\noindent Proof: The eigenvalue equation (\ref{eq:eigenvalue_eqn_g_tilde}) reads, upon expanding the second row in terms of the adjacency matrix $\mathbf{A}$,

\begin{eqnarray}
 g (\tilde{\mathbf{x}}_0)^k_a +  \sum_{j}{\mathbf{A}}_{2j} (\tilde{\mathbf{x}}_0)^k_j  = {0},~~~ \text{for all}~~~ k.   
\end{eqnarray}
A linear combination with real coefficients $b_k$ of the following form can be written  

\begin{eqnarray}
&& \sum_k  b_k (\tilde{\mathbf{x}}_0)^k_j = 0, \nonumber\\
&& \implies \sum_j \mathbf{A}_{2j}(\sum_k  b_k(\tilde{\mathbf{x}}_0)^k_j) = 0.
\end{eqnarray}

Therefore, the vector $\sum_k  b_k(\tilde{\mathbf{x}}_c)^k_j$, is  an eigenvector of $\mathbf{A}$ corresponding to the zero eigenvalue  that is spanned by the eigenvector(s) $\tilde{\mathbf{x}}_0$ of the null space of $\tilde{\mathbf{A}}$.\\

\subsection*{4. AO symmetry of graphs with self-loops}

{\color{black} In the adjacency-graph representation, a local transverse-field along z, which constitutes an on-site potential, may correspond to self-loops with weights given by the field strength~\cite{Mulken2011_CTWQ_review,Babai2016}. 
The symmetry classification by automorphism orbits (AOs) can then be generalized by the theory of vertex-automorphism groups to self-loop augmented matrix~\cite{Babai2016,Grohe2018}
\begin{equation}
\tilde{A}_{kl} =
\begin{cases}
g_{kl}, & k\neq l,\\[3pt]
h_k, & k=l, 
\end{cases}
\end{equation}
The automorphism group $\mathrm{Aut}(\tilde{A})$ is defined as the group of all permutations $\pi$ satisfying 
\begin{equation}
    \tilde{A}_{\pi(k)\pi(l)}=\tilde{A}_{kl} \quad \text{for all} \quad k,l. 
\end{equation}
If, for all nodes, the degeneracy is lifted by local magnetic fields, then $\mathrm{Aut}(\tilde{A})$ yields
\begin{equation}
 K=N,    
\end{equation}
and the RT protocol guarantees complete purification. Thus, local inhomogeneous fields lift residual AO degeneracies and enhance the purification bound.
A plot to demonstrate the effect of the on-site magnetic field (Fig.~\ref{fig:9a}). 
}

\section*{Acknowledgments}
G.K. and  D.B.R.D  would like to acknowledge support from DFG (Project no. FOR2724). D.B.R.D would like to acknowledge support from the German Ministry of Education and Research for the project QECHQS (BMBF, Grant agreement no. 16KIS1590K). A.F. is the incumbent of the Elaine Blond Career Development Chair. \"{O}. E. M. acknowledges the support by the Scientific and Technological
Research Council of Türkiye (TUBITAK) under grant number 123F150. N.E.P. acknowledges the support of the HORIZON-RIA Project EuRyQa (grant No. 101070144).  P.C. acknowledges the support from the International Postdoctoral Fellowship from the Ben May Center for Theory and Computation. We acknowledge discussions with A.~\"{O}zdemir and S. Chakraborty.



\vspace{3ex}

\textbf{Author contributions:} S.S., P.C., A.C., N.P.,  D.D.B.R. did  the calculations. S.S., A.C., A.F., D.D.B.R. prepared the figures. G.K., O.E.M, and D.D.B.R. conceptualized the problem. All authors wrote and reviewed the manuscript.

\vspace{3ex}
\textbf{Competing interests:} The authors declare no competing interests.

\vspace{4ex}
\bibliography{manuscript}

\begin{thebibliography}{22}%
\makeatletter
\providecommand \@ifxundefined [1]{%
 \@ifx{#1\undefined}
}%
\providecommand \@ifnum [1]{%
 \ifnum #1\expandafter \@firstoftwo
 \else \expandafter \@secondoftwo
 \fi
}%
\providecommand \@ifx [1]{%
 \ifx #1\expandafter \@firstoftwo
 \else \expandafter \@secondoftwo
 \fi
}%
\providecommand \natexlab [1]{#1}%
\providecommand \enquote  [1]{``#1''}%
\providecommand \bibnamefont  [1]{#1}%
\providecommand \bibfnamefont [1]{#1}%
\providecommand \citenamefont [1]{#1}%
\providecommand \href@noop [0]{\@secondoftwo}%
\providecommand \href [0]{\begingroup \@sanitize@url \@href}%
\providecommand \@href[1]{\@@startlink{#1}\@@href}%
\providecommand \@@href[1]{\endgroup#1\@@endlink}%
\providecommand \@sanitize@url [0]{\catcode `\\12\catcode `\$12\catcode `\&12\catcode `\#12\catcode `\^12\catcode `\_12\catcode `\%12\relax}%
\providecommand \@@startlink[1]{}%
\providecommand \@@endlink[0]{}%
\providecommand \url  [0]{\begingroup\@sanitize@url \@url }%
\providecommand \@url [1]{\endgroup\@href {#1}{\urlprefix }}%
\providecommand \urlprefix  [0]{URL }%
\providecommand \Eprint [0]{\href }%
\providecommand \doibase [0]{https://doi.org/}%
\providecommand \selectlanguage [0]{\@gobble}%
\providecommand \bibinfo  [0]{\@secondoftwo}%
\providecommand \bibfield  [0]{\@secondoftwo}%
\providecommand \translation [1]{[#1]}%
\providecommand \BibitemOpen [0]{}%
\providecommand \bibitemStop [0]{}%
\providecommand \bibitemNoStop [0]{.\EOS\space}%
\providecommand \EOS [0]{\spacefactor3000\relax}%
\providecommand \BibitemShut  [1]{\csname bibitem#1\endcsname}%
\let\auto@bib@innerbib\@empty
\bibitem [{\citenamefont {Dehmer}\ and\ \citenamefont {(ed.)}(2009)}]{dehmer}%
  \BibitemOpen
  \bibfield  {author} {\bibinfo {author} {\bibfnamefont {M.}~\bibnamefont {Dehmer}}\ and\ \bibinfo {author} {\bibfnamefont {F.~E.-S.}\ \bibnamefont {(ed.)}},\ }\href@noop {} {\emph {\bibinfo {title} {Analysis of Complex Networks From Biology to Linguistics}}}\ (\bibinfo  {publisher} {WILEY-VCH Verlag},\ \bibinfo {year} {2009})\BibitemShut {NoStop}%
\bibitem [{\citenamefont {Pr{\v z}ulj}(2007)}]{przulj_2007}%
  \BibitemOpen
  \bibfield  {author} {\bibinfo {author} {\bibfnamefont {N.}~\bibnamefont {Pr{\v z}ulj}},\ }\bibfield  {title} {\bibinfo {title} {{Biological network comparison using graphlet degree distribution}},\ }\href {https://doi.org/10.1093/bioinformatics/btl301} {\bibfield  {journal} {\bibinfo  {journal} {Bioinformatics}\ }\textbf {\bibinfo {volume} {23}},\ \bibinfo {pages} {e177} (\bibinfo {year} {2007})}\BibitemShut {NoStop}%
\bibitem [{\citenamefont {Finotelli~Paolo}\ and\ \citenamefont {Paolo}(2021)}]{paolo_2021}%
  \BibitemOpen
  \bibfield  {author} {\bibinfo {author} {\bibfnamefont {M.~E.}\ \bibnamefont {Finotelli~Paolo}, \bibfnamefont {Piccardi~Carlo}}\ and\ \bibinfo {author} {\bibfnamefont {D.}~\bibnamefont {Paolo}},\ }\bibfield  {title} {\bibinfo {title} {{A Graphlet-Based Topological Characterization of the Resting-State Network in Healthy People}},\ }\bibfield  {journal} {\bibinfo  {journal} {Frontiers in Neuroscience}\ }\textbf {\bibinfo {volume} {15}},\ \href {https://doi.org/10.3389/fnins.2021.665544} {10.3389/fnins.2021.665544} (\bibinfo {year} {2021})\BibitemShut {NoStop}%
\bibitem [{\citenamefont {Mowshowitz}\ \emph {et~al.}(2019)\citenamefont {Mowshowitz}, \citenamefont {Dehmer},\ and\ \citenamefont {Emmert-Streib}}]{mowshowitz}%
  \BibitemOpen
  \bibfield  {author} {\bibinfo {author} {\bibfnamefont {A.}~\bibnamefont {Mowshowitz}}, \bibinfo {author} {\bibfnamefont {M.}~\bibnamefont {Dehmer}},\ and\ \bibinfo {author} {\bibfnamefont {F.}~\bibnamefont {Emmert-Streib}},\ }\bibfield  {title} {\bibinfo {title} {A note on graphs with prescribed orbit structure},\ }\bibfield  {journal} {\bibinfo  {journal} {Entropy}\ }\textbf {\bibinfo {volume} {21}},\ \href {https://doi.org/10.3390/e21111118} {10.3390/e21111118} (\bibinfo {year} {2019})\BibitemShut {NoStop}%
\bibitem [{\citenamefont {Meyer}(2000)}]{meyer01}%
  \BibitemOpen
  \bibfield  {author} {\bibinfo {author} {\bibfnamefont {C.~D.}\ \bibnamefont {Meyer}},\ }\href@noop {} {\emph {\bibinfo {title} {Matrix Analysis and Applied Linear Algebra}}}\ (\bibinfo  {publisher} {SIAM},\ \bibinfo {year} {2000})\BibitemShut {NoStop}%
\bibitem [{\citenamefont {Askey}(1975)}]{askey_book}%
  \BibitemOpen
  \bibfield  {author} {\bibinfo {author} {\bibfnamefont {R.}~\bibnamefont {Askey}},\ }\href@noop {} {\emph {\bibinfo {title} {Orthogonal polynomials and special functions}}}\ (\bibinfo  {publisher} {SIAM},\ \bibinfo {year} {1975})\BibitemShut {NoStop}%
\bibitem [{\citenamefont {Karamata}(1932)}]{karamata_1932}%
  \BibitemOpen
  \bibfield  {author} {\bibinfo {author} {\bibfnamefont {J.}~\bibnamefont {Karamata}},\ }\bibfield  {title} {\bibinfo {title} {Sur une in{\'e}galit{\'e} r{\'e}lative aux fonctions convexes},\ }\href@noop {} {\bibfield  {journal} {\bibinfo  {journal} {Publ. Math. Univ. Belgrade}\ }\textbf {\bibinfo {volume} {1}},\ \bibinfo {pages} {145} (\bibinfo {year} {1932})}\BibitemShut {NoStop}%
\bibitem [{\citenamefont {Bethe}(1931)}]{bethe}%
  \BibitemOpen
  \bibfield  {author} {\bibinfo {author} {\bibfnamefont {H.}~\bibnamefont {Bethe}},\ }\bibfield  {title} {\bibinfo {title} {Zur theorie der metalle},\ }\href {https://doi.org/10.1007/BF01341708} {\bibfield  {journal} {\bibinfo  {journal} {Zeitschrift f{\"u}r Physik}\ }\textbf {\bibinfo {volume} {71}},\ \bibinfo {pages} {205} (\bibinfo {year} {1931})}\BibitemShut {NoStop}%
\bibitem [{\citenamefont {Izyumov}\ and\ \citenamefont {Skryabin}(1988)}]{izyumov}%
  \BibitemOpen
  \bibfield  {author} {\bibinfo {author} {\bibfnamefont {Y.~A.}\ \bibnamefont {Izyumov}}\ and\ \bibinfo {author} {\bibfnamefont {Y.~N.}\ \bibnamefont {Skryabin}},\ }\href@noop {} {\emph {\bibinfo {title} {Statistical Mechanics of Magnetically Ordered Systems}}}\ (\bibinfo  {publisher} {Springer-Verlag, Heidelberg},\ \bibinfo {year} {1988})\BibitemShut {NoStop}%
\bibitem [{\citenamefont {ZACHOS}(1992)}]{zachos_1992}%
  \BibitemOpen
  \bibfield  {author} {\bibinfo {author} {\bibfnamefont {C.}~\bibnamefont {ZACHOS}},\ }\bibfield  {title} {\bibinfo {title} {Altering the symmetry of wave functions in quantum algebras and supersymmetry},\ }\href {https://doi.org/10.1142/S0217732392001270} {\bibfield  {journal} {\bibinfo  {journal} {Modern Physics Letters A}\ }\textbf {\bibinfo {volume} {07}},\ \bibinfo {pages} {1595} (\bibinfo {year} {1992})},\ \Eprint {https://arxiv.org/abs/https://doi.org/10.1142/S0217732392001270} {https://doi.org/10.1142/S0217732392001270} \BibitemShut {NoStop}%
\bibitem [{\citenamefont {Sakurai}(1994)}]{sakurai_book}%
  \BibitemOpen
  \bibfield  {author} {\bibinfo {author} {\bibfnamefont {J.~J.}\ \bibnamefont {Sakurai}},\ }\href {https://cds.cern.ch/record/1167961} {\emph {\bibinfo {title} {{Modern quantum mechanics; rev. ed.}}}}\ (\bibinfo  {publisher} {Addison-Wesley},\ \bibinfo {address} {Reading, MA},\ \bibinfo {year} {1994})\BibitemShut {NoStop}%
\bibitem [{\citenamefont {\'Alvarez}\ \emph {et~al.}(2010)\citenamefont {\'Alvarez}, \citenamefont {Rao}, \citenamefont {Frydman},\ and\ \citenamefont {Kurizki}}]{alvarez_2010}%
  \BibitemOpen
  \bibfield  {author} {\bibinfo {author} {\bibfnamefont {G.~A.}\ \bibnamefont {\'Alvarez}}, \bibinfo {author} {\bibfnamefont {D.~D.~B.}\ \bibnamefont {Rao}}, \bibinfo {author} {\bibfnamefont {L.}~\bibnamefont {Frydman}},\ and\ \bibinfo {author} {\bibfnamefont {G.}~\bibnamefont {Kurizki}},\ }\bibfield  {title} {\bibinfo {title} {Zeno and anti-zeno polarization control of spin ensembles by induced dephasing},\ }\href {https://doi.org/10.1103/PhysRevLett.105.160401} {\bibfield  {journal} {\bibinfo  {journal} {Phys. Rev. Lett.}\ }\textbf {\bibinfo {volume} {105}},\ \bibinfo {pages} {160401} (\bibinfo {year} {2010})}\BibitemShut {NoStop}%
\bibitem [{\citenamefont {Ranzani}\ and\ \citenamefont {Aumentado}(2015)}]{RanzaniAumentado2015}%
  \BibitemOpen
  \bibfield  {author} {\bibinfo {author} {\bibfnamefont {L.}~\bibnamefont {Ranzani}}\ and\ \bibinfo {author} {\bibfnamefont {J.}~\bibnamefont {Aumentado}},\ }\bibfield  {title} {\bibinfo {title} {Graph-based analysis of nonreciprocity in coupled-mode systems},\ }\href {https://doi.org/10.1088/1367-2630/17/2/023024} {\bibfield  {journal} {\bibinfo  {journal} {New Journal of Physics}\ }\textbf {\bibinfo {volume} {17}},\ \bibinfo {pages} {023024} (\bibinfo {year} {2015})}\BibitemShut {NoStop}%
\bibitem [{\citenamefont {McKay}\ \emph {et~al.}(2016)\citenamefont {McKay}, \citenamefont {Filipp}, \citenamefont {Mezzacapo}, \citenamefont {Magesan}, \citenamefont {Chow},\ and\ \citenamefont {Gambetta}}]{McKay2016}%
  \BibitemOpen
  \bibfield  {author} {\bibinfo {author} {\bibfnamefont {D.~C.}\ \bibnamefont {McKay}}, \bibinfo {author} {\bibfnamefont {S.}~\bibnamefont {Filipp}}, \bibinfo {author} {\bibfnamefont {A.}~\bibnamefont {Mezzacapo}}, \bibinfo {author} {\bibfnamefont {E.}~\bibnamefont {Magesan}}, \bibinfo {author} {\bibfnamefont {J.~M.}\ \bibnamefont {Chow}},\ and\ \bibinfo {author} {\bibfnamefont {J.~M.}\ \bibnamefont {Gambetta}},\ }\bibfield  {title} {\bibinfo {title} {Universal gate for fixed-frequency qubits via a tunable bus},\ }\href {https://doi.org/10.1103/PhysRevApplied.6.064007} {\bibfield  {journal} {\bibinfo  {journal} {Physical Review Applied}\ }\textbf {\bibinfo {volume} {6}},\ \bibinfo {pages} {064007} (\bibinfo {year} {2016})}\BibitemShut {NoStop}%
\bibitem [{\citenamefont {Stehlik}\ \emph {et~al.}(2021)\citenamefont {Stehlik}, \citenamefont {Zajac}, \citenamefont {Underwood}, \citenamefont {Phung}, \citenamefont {Blair}, \citenamefont {Carnevale}, \citenamefont {Klaus}, \citenamefont {Keefe}, \citenamefont {Carniol}, \citenamefont {Kumph}, \citenamefont {Steffen},\ and\ \citenamefont {Dial}}]{PhysRevLett.127.080505}%
  \BibitemOpen
  \bibfield  {author} {\bibinfo {author} {\bibfnamefont {J.}~\bibnamefont {Stehlik}}, \bibinfo {author} {\bibfnamefont {D.~M.}\ \bibnamefont {Zajac}}, \bibinfo {author} {\bibfnamefont {D.~L.}\ \bibnamefont {Underwood}}, \bibinfo {author} {\bibfnamefont {T.}~\bibnamefont {Phung}}, \bibinfo {author} {\bibfnamefont {J.}~\bibnamefont {Blair}}, \bibinfo {author} {\bibfnamefont {S.}~\bibnamefont {Carnevale}}, \bibinfo {author} {\bibfnamefont {D.}~\bibnamefont {Klaus}}, \bibinfo {author} {\bibfnamefont {G.~A.}\ \bibnamefont {Keefe}}, \bibinfo {author} {\bibfnamefont {A.}~\bibnamefont {Carniol}}, \bibinfo {author} {\bibfnamefont {M.}~\bibnamefont {Kumph}}, \bibinfo {author} {\bibfnamefont {M.}~\bibnamefont {Steffen}},\ and\ \bibinfo {author} {\bibfnamefont {O.~E.}\ \bibnamefont {Dial}},\ }\bibfield  {title} {\bibinfo {title} {Tunable coupling architecture for fixed-frequency transmon superconducting qubits},\ }\href {https://doi.org/10.1103/PhysRevLett.127.080505} {\bibfield  {journal} {\bibinfo  {journal} {Phys.
  Rev. Lett.}\ }\textbf {\bibinfo {volume} {127}},\ \bibinfo {pages} {080505} (\bibinfo {year} {2021})}\BibitemShut {NoStop}%
\bibitem [{\citenamefont {Kamal}\ \emph {et~al.}(2011)\citenamefont {Kamal}, \citenamefont {Clarke},\ and\ \citenamefont {Devoret}}]{Kamal2011}%
  \BibitemOpen
  \bibfield  {author} {\bibinfo {author} {\bibfnamefont {A.}~\bibnamefont {Kamal}}, \bibinfo {author} {\bibfnamefont {J.}~\bibnamefont {Clarke}},\ and\ \bibinfo {author} {\bibfnamefont {M.~H.}\ \bibnamefont {Devoret}},\ }\bibfield  {title} {\bibinfo {title} {Noiseless nonreciprocity in a parametric active device},\ }\href {https://doi.org/10.1038/nphys1893} {\bibfield  {journal} {\bibinfo  {journal} {Nature Physics}\ }\textbf {\bibinfo {volume} {7}},\ \bibinfo {pages} {311} (\bibinfo {year} {2011})}\BibitemShut {NoStop}%
\bibitem [{\citenamefont {Xu}\ \emph {et~al.}(2020)\citenamefont {Xu} \emph {et~al.}}]{Xu2020}%
  \BibitemOpen
  \bibfield  {author} {\bibinfo {author} {\bibfnamefont {M.}~\bibnamefont {Xu}} \emph {et~al.},\ }\bibfield  {title} {\bibinfo {title} {Directional photon exchange in a superconducting circuit--quantum electrodynamics architecture},\ }\href {https://doi.org/10.1038/s41567-020-0910-y} {\bibfield  {journal} {\bibinfo  {journal} {Nature Physics}\ }\textbf {\bibinfo {volume} {16}},\ \bibinfo {pages} {911} (\bibinfo {year} {2020})}\BibitemShut {NoStop}%
\bibitem [{\citenamefont {Chapman}\ \emph {et~al.}(2017)\citenamefont {Chapman} \emph {et~al.}}]{Chapman2017}%
  \BibitemOpen
  \bibfield  {author} {\bibinfo {author} {\bibfnamefont {B.~J.}\ \bibnamefont {Chapman}} \emph {et~al.},\ }\bibfield  {title} {\bibinfo {title} {Widely tunable on-chip microwave circulator for superconducting quantum circuits},\ }\href {https://doi.org/10.1103/PhysRevX.7.041043} {\bibfield  {journal} {\bibinfo  {journal} {Physical Review X}\ }\textbf {\bibinfo {volume} {7}},\ \bibinfo {pages} {041043} (\bibinfo {year} {2017})}\BibitemShut {NoStop}%
\bibitem [{\citenamefont {Maffei}\ \emph {et~al.}(2024)\citenamefont {Maffei}, \citenamefont {Pomarico}, \citenamefont {Facchi}, \citenamefont {Magnifico}, \citenamefont {Pascazio},\ and\ \citenamefont {Pepe}}]{PhysRevResearch.6.L032017}%
  \BibitemOpen
  \bibfield  {author} {\bibinfo {author} {\bibfnamefont {M.}~\bibnamefont {Maffei}}, \bibinfo {author} {\bibfnamefont {D.}~\bibnamefont {Pomarico}}, \bibinfo {author} {\bibfnamefont {P.}~\bibnamefont {Facchi}}, \bibinfo {author} {\bibfnamefont {G.}~\bibnamefont {Magnifico}}, \bibinfo {author} {\bibfnamefont {S.}~\bibnamefont {Pascazio}},\ and\ \bibinfo {author} {\bibfnamefont {F.~V.}\ \bibnamefont {Pepe}},\ }\bibfield  {title} {\bibinfo {title} {Directional emission and photon bunching from a qubit pair in waveguide},\ }\href {https://doi.org/10.1103/PhysRevResearch.6.L032017} {\bibfield  {journal} {\bibinfo  {journal} {Phys. Rev. Res.}\ }\textbf {\bibinfo {volume} {6}},\ \bibinfo {pages} {L032017} (\bibinfo {year} {2024})}\BibitemShut {NoStop}%
\bibitem [{\citenamefont {Metelmann}\ and\ \citenamefont {Clerk}(2015)}]{MetelmannClerk2015}%
  \BibitemOpen
  \bibfield  {author} {\bibinfo {author} {\bibfnamefont {A.}~\bibnamefont {Metelmann}}\ and\ \bibinfo {author} {\bibfnamefont {A.~A.}\ \bibnamefont {Clerk}},\ }\bibfield  {title} {\bibinfo {title} {Nonreciprocal photon transmission and amplification via reservoir engineering},\ }\href {https://doi.org/10.1103/PhysRevX.5.021025} {\bibfield  {journal} {\bibinfo  {journal} {Physical Review X}\ }\textbf {\bibinfo {volume} {5}},\ \bibinfo {pages} {021025} (\bibinfo {year} {2015})}\BibitemShut {NoStop}%
\bibitem [{\citenamefont {Yu}\ \emph {et~al.}(2023)\citenamefont {Yu}, \citenamefont {Luo},\ and\ \citenamefont {Bauer}}]{Yu2023ChiralitySpintronics}%
  \BibitemOpen
  \bibfield  {author} {\bibinfo {author} {\bibfnamefont {T.}~\bibnamefont {Yu}}, \bibinfo {author} {\bibfnamefont {Z.}~\bibnamefont {Luo}},\ and\ \bibinfo {author} {\bibfnamefont {G.~E.~W.}\ \bibnamefont {Bauer}},\ }\bibfield  {title} {\bibinfo {title} {Chirality as generalized spin–orbit interaction in spintronics},\ }\href {https://doi.org/10.1016/j.physrep.2023.01.002} {\bibfield  {journal} {\bibinfo  {journal} {Physics Reports}\ }\textbf {\bibinfo {volume} {1009}},\ \bibinfo {pages} {1} (\bibinfo {year} {2023})}\BibitemShut {NoStop}%
\bibitem [{\citenamefont {Geng}\ \emph {et~al.}(2023)\citenamefont {Geng}, \citenamefont {Wei}, \citenamefont {Wang},\ and\ \citenamefont {Xing}}]{Geng2023}%
  \BibitemOpen
  \bibfield  {author} {\bibinfo {author} {\bibfnamefont {H.}~\bibnamefont {Geng}}, \bibinfo {author} {\bibfnamefont {J.}~\bibnamefont {Wei}}, \bibinfo {author} {\bibfnamefont {Z.}~\bibnamefont {Wang}},\ and\ \bibinfo {author} {\bibfnamefont {D.}~\bibnamefont {Xing}},\ }\bibfield  {title} {\bibinfo {title} {Nonreciprocal charge and spin transport induced by non-hermitian skin effect in mesoscopic systems},\ }\href {https://doi.org/10.1103/PhysRevB.107.035306} {\bibfield  {journal} {\bibinfo  {journal} {Physical Review B}\ }\textbf {\bibinfo {volume} {107}},\ \bibinfo {pages} {035306} (\bibinfo {year} {2023})}\BibitemShut {NoStop}%
\end{thebibliography}

\clearpage
\onecolumngrid
\setcounter{equation}{0}
\renewcommand{\theequation}{S\arabic{equation}}

\begin{center}
    {\textbf{Supplementary Material: Collective purification of interacting quantum networks beyond symmetry constraints}} \\[0.5cm]
\end{center}

\section{Matrix representation of $\mathcal{M}$}\label{appendix:A}

Let the basis states for the $i$th spin be $\vert 0\rangle$ (down-spin state) and  $\vert 1\rangle$ (up-spin state), denoting the eigenstates of $\sigma_i^z$ with eigenvalues $-1$ and $+1$ respectively. The $z$-component of the total spin $\sum_i \sigma_i^z$ is a constant of motion,  which implies that the eigenstates will have a definite number of up spins and parity. 

The many-qubit basis states with $l$ down spins can 
be labeled by the locations $\mathbf{x}=(x_1,x_2,...,x_l$), in an ordered set with $x_1<x_2...<x_l$. An eigenstate  can be written as a superposition of the basis states,
\begin{equation}
	\vert\psi\rangle = \sum_{\mathbf{x}}\Phi^q_{\mathbf{x}} \vert \mathbf{x}\rangle,
\end{equation}
where $\Phi^q_{\mathbf{x}}$ denotes the amplitude of the  basis state $\vert \mathbf{x}\rangle$ that is labeled by the quantum number $q$. The matrix elements of the time evolution operator in the computational basis  with the eigenenergies $E_q$. are given as
\begin{eqnarray}
	\hat{U} (\mathbf{x},\mathbf{x}')= \sum_{\mathbf{x},\mathbf{x}'} \sum_q \Phi^{q}_{\mathbf{x}} \Phi^{*q}_{\mathbf{x}'} e^{-i\tau E_q}|\mathbf{x}\rangle \langle \mathbf{x}'|.
\end{eqnarray}

The matrix form of the map in Eq.~(2b) (main text) can be represented by a set of $2^{N+1}$ linear homogeneous algebraic equations with complex coefficients in the following form
\begin{eqnarray}
	\sum_{\alpha,\beta} \hat{M}^{\alpha,\beta}_{i,j} \rho^S_{\alpha,\beta} = 0,
	\label{eq:matrix_eqn}
\end{eqnarray}
 with the matrix elements given by the following formula
\begin{eqnarray}
	\hat{M}^{\alpha,\beta}_{i,j} = (\hat{U} (i,\alpha) \hat{U}^\dag (\beta,j) + \hat{U} (i',\alpha) \hat{U}^\dag (\beta,j') ) - \delta_{i,\alpha} \delta_{j,\beta}, \nonumber\\
\end{eqnarray}
where $\hat{U} (i,\alpha)$ and $\hat{U}^\dag (\beta,j)$ are each labeled by two location sets and $|i'\rangle = \hat{\sigma}^x_A |i\rangle$, $|j'\rangle = \hat{\sigma}^x_A |j\rangle$ . The above set of linear equations can be transformed to a set of $4^N-1$ linear equations with real coefficients using the Hermitian and trace-preserving conditions of density matrices. 

\setcounter{equation}{0}
\renewcommand{\theequation}{B\arabic{equation}}
\section{Analytical formulae for $\mathscr{N}(\mathbcal{M})$ and $\mathscr{R}(\mathbcal{M})$}\label{appendix:B}

\textcolor{black} A graph $G(V,E)$ is a binary relation over a finite set of vertices $V$ with a finite set of edges $E$. Here, an edge $e \in E$ is a spinor $e = (u,v)$ for a pair of adjacent vertices $u,v \in V; u \ne v$. The graphs of an dipolar-interacting spin network have no loops, i.e., edges joining the same vertices. These graphs are known as weighted graph.

\textit{\textcolor{black} {a) A graph automorphism of a graph  $G$ is a permutation $\phi$ from the graph $G$ to itself on the vertices $V_G$ that satisfies,}
$$f: ~V(G) \rightarrow V(G)~ \mbox{such that} ~ uv \in E(G)~ \mbox{iff}~ \phi(u)\phi(v) \in E(G).$$}

\textcolor{black} {The automorphism group $Aut(G) $ denotes the set of all automorphisms on  graph $G$, which forms a group under the function composition~\cite{Diestel_springer_2006_1,dehmer}. Therefore the set  $Aut(G)$ is non-empty and contains at least the identity.}\\

\textcolor{black} {b) If $x$ is a node of a graph $G$, the automorphism orbit (AO) of $x$ is $Orb(x) = \{ y \in V(x)| y = \pi(x)$ for the exchange permutation operation $ \hat{\pi} \in Aut(G)\}$. ~\cite{przulj_2007, paolo_2021, dehmer,mowshowitz}}

\textcolor{black} {The orbit of an automorphism group (AO) of a graph constitutes a decomposition of the vertices of the graph that determines the complexity of the graph. }

Let us have a graph of $N$ nodes with $K$ distinct automorphism orbits and assume that $n_j$ is the number of the nodes belonging to $j$-th orbit, which are labeled with the same color. The number of $m$ spin density matrices will be equal to the number of possible ways $S_m(n_j)$ one can choose $m$ nodes distinctly. Note that more than one nodes with same orbit can be chosen. 

If each node of the graph is distinct from the rest, the number of possible $m$ qubit density matrices would have been ${N \choose m}$. The number of possible ways $S_m$ one can choose $m$ nodes distinctly is 
\begin{eqnarray}
S_m(\{n_j\}) = {N \choose m} - D_m(\{n_j\}),
\end{eqnarray}
where $D_m(\{n_j\})$ is the extra counting due to color degeneracy, which is to be determined
subsequently.

If $n_j$ is the number of nodes in orbit $j$, the possible number of ways one can choose an $m$ spin density matrix with a single color is $\sum^K_{j=1} {{n_j} \choose m}$. We subtract this number from $S_m(\{n_j\})$, except one for each color. Hence, the first term in $D_m(\{n_j\})$ will be 
\begin{eqnarray}
\sum^{K}_{j=1}  f\big({{n_j} \choose m} -1\big),
\end{eqnarray}
with  the positive counting  function $f(y) = (y +|y|)/2$.
 
The number of identical density matrices for  $m$ spins with two distinct colors will be $$\sum^K_{j_1,j_2; j_2 \ge j_1} \sum^m_{p_1,p_2 \ge 1,p_1+p_2=m}  {{n_{j_1}} \choose {p_1}} {{n_{j_2}} \choose {p_2}}.$$
We subtract this number from $S_m(\{n_j\})$ except one for each bi-color combination. Hence, the second term in $D_m(\{n_j\})$ will be 
\begin{eqnarray}
\sum^K_{j_1,j_2; j_2 \ge j_1} \sum^m_{p_1,p_2 \ge 1,p_1+p_2=m} f({{n_{j_1}} \choose {p_1}} {{n_{j_2}} \choose {p_2}} -1).
\end{eqnarray}

The general form of $D_m(\{n_j\})$ is given by
\begin{eqnarray}
D_m(\{n_j\}) =  \sum^{m}_{\alpha = 1} \sum^K_{j_{\alpha} =1} \sum^N_{\{ p_{\alpha}\} = 1, \sum p_\alpha = m}  f\big(\prod^N_{\alpha=1} {{n_{{\alpha}}} \choose {p_{\alpha}}}-1\big) \nonumber\\
\end{eqnarray}
Therefore, the rank $\mathscr{R}(\mathcal{M})$ and the nullity $\mathscr{N}(\mathcal{M})$~\cite{meyer01} are expressed by
\begin{eqnarray}
\mathscr{R}(\mathbcal{M}) &=& \sum^N_{m=1} 3^m S_m(\{n_j\})\nonumber\\
&=&  (4^N-1) - \sum^N_{m=1} 3^m D_m(\{n_j\}), \\
\mathscr{N}(\mathbcal{M}) &=& \sum^N_{m=1} 3^m D_m(\{n_j\}).
\end{eqnarray}
The size of the steady-state 
subspace (SSS) is bounded by
\begin{eqnarray}
 0 \le \mathscr{N}(\mathbcal{M})_G \le (4^N-1) - (2 \times 3^N-3),
 \label{upper_bound}
\end{eqnarray}
where the identity graph and the complete graph satisfy the lower and upper bounds, respectively.

\setcounter{equation}{0}
\renewcommand{\theequation}{C\arabic{equation}}
\section{Mapping between nullity and node distribution in automorphism orbits }\label{appendix:C}

Let us assume two distinct graphs $G$ and $G'$ with node distributions $(n_1,n_2,...n_N)$ and $(n'_1,n'_2,...n'_N)$. Without any loss of generality the node distribution in $G$ majorizes that of $G'$, i.e. $(n_1,n_2,...n_N) \succ (n'_1,n'_2,...n'_N)$ if,
\begin{eqnarray}
\sum^N_{j=1} n_j = \sum^N_{j=1} n'_j = N, \nonumber\\
\sum^i_{j=1} n_j > \sum^i_{j=1} n'_j, i \le N.
\end{eqnarray}

From Chu-Vandermonde identity~\cite{askey_book} we have
\begin{eqnarray}
\sum^N_{\{ p_{\alpha}\} = 1}  \prod^N_{\alpha} {{n_{{\alpha}}} \choose {p_{\alpha}}} = \sum^N_{\{ p_{\alpha}\} = 1}  \prod^N_{\alpha} {{n'_{{\alpha}}} \choose {p_{\alpha}}} = 
{N \choose m},
\end{eqnarray}
with the constraint $\sum^N_{\alpha = 1} p_\alpha = m$. 
Therefore,
\begin{eqnarray}
\sum^N_{\{ p_{\alpha}\} = 1}  f\big(\prod^N_{\alpha} {{n_{{\alpha}}} \choose {p_{\alpha}}}-1\big) > \sum^N_{\{ p_{\alpha}\} = 1}  f\big(\prod^N_{\alpha} {{n'_{{\alpha}}} \choose {p_{\alpha}}}-1\big), \nonumber\\
\end{eqnarray}
for at least one $m$, which implies 
\begin{eqnarray}
D_m(\{n_j\}) &>& D_m(\{n'_j\})\nonumber\\
 \implies 
 \mathscr{N}(\mathbcal{M})_G &>& \mathscr{N}(\mathbcal{M})_G'.
\end{eqnarray}
This proves $\mathscr{N}(\mathbcal{M})$ is a strictly convex function. Then from Karamata inequality~\cite{karamata_1932}, we conclude that $\mathscr{N}(\mathbcal{M})$ is a one-to-one mapping of $(n_1,n_2,...n_N)$.

\setcounter{equation}{0}
\renewcommand{\theequation}{D\arabic{equation}}
\section{System-ancilla interaction Hamiltonian for purification of non interacting spins}\label{appendix:D}

 The Hamiltonian (1a) in the lab-frame,
  \begin{eqnarray}
\hat{H}^{Lab} (t) = \omega_0 \hat{\sigma}^z_A + \Omega_0 \big[ \hat{\sigma}^+_A e^{-i\omega_0t} + \hat{\sigma}^-_A e^{i\omega_0t}\big]  + \omega_L \sum_k \hat{\sigma}^z_k + \hat{\sigma}^z \sum_k \vec{g}_k. \hat{\vec{\sigma}}_k.
\end{eqnarray}
In a frame rotating about the $z$-axis of the ancilla with the phase and frequency of its Larmor precession,
  \begin{eqnarray}
\hat{H}^{Rot} = \Omega_0 \hat{\sigma}^x_A +  \omega_L \sum_k \hat{\sigma}^z_k + \hat{\sigma}^z \sum_k \vec{g}_k. {\vec{\sigma}_k}.
\end{eqnarray}  
Expressing the above Hamiltonian in a dressed basis of the ancilla that transforms $\hat{\sigma}^x_A \leftrightarrow {} \hat{\sigma}^z_A$, we can define,


$\hat{H}_0 = \Omega_0 \hat{\sigma}^z_A + \omega_L \sum_k \hat{\sigma}^z_k$, $\hat{H}_{SA} = \hat{\sigma}^x_A \sum_k \vec{g}_k.\vec{\sigma}_k$, $\hat{H}^{D} = \hat{H}_0 + H_{SA}$. 

In the interaction representation w.r.t $\hat{H}_0$,

  \begin{eqnarray}
\hat{H}_{SA}(t) = e^{-i\hat{H}_0t} \hat{H}_{SA} e^{i\hat{H}_0t}.
\end{eqnarray}

Assuming that the $S-A$ coupling Hamiltonian is restricted to the X-Z plane of the chain spins, i.e., $\vec{g}_k. \vec{I}_k = g^x_k I^x_k + g^z_k I^z_k$, it becomes,

  \begin{eqnarray}
\hat{H}(t) = \sum_k g^x_k \big[ e^{-i(\Omega_0 - \omega_L)t} \hat{\sigma}^+_A\hat{\sigma}^-_k + e^{-i(\Omega_0 + \omega_L)t}\hat{\sigma}^+_A \hat{\sigma}^+_k + h.c.  \big]
+\sum_k g^z_k \big[ e^{-i\Omega_0t}\hat{\sigma}^+_A +h.c.\big] \hat{\sigma}^z_k. \nonumber\\
\end{eqnarray}

Neglecting the non-secular contributions and setting the resonance condition $\Omega_0 = \omega_L$, we get

  \begin{eqnarray}
\hat H_{SA} =\hat H^{res}_{SA}  = \sum^N_{k=1} g_k(\hat{\sigma}^+_A \hat{\sigma}^-_k +\hat{\sigma}^-_A \hat{\sigma}^+_k).
\label{eq: res}
\end{eqnarray}

Let us consider, for simplicity, that each spin is coupled to the ancilla with the same coupling strength $g$ (star model). Then from Eq.~\eqref{eq: res} becomes

  \begin{eqnarray}
\hat{H}^{res}_{SA}  =  \frac{g}{2}  \sum_k (\hat{\sigma}^+_A \hat{\sigma}^-_k +\hat{\sigma}^-_A \hat{\sigma}^+_k).
\end{eqnarray}




The corresponding time-evolution operator can then be simplified as
\begin{eqnarray}\nonumber
\hat U^{res}(t) = \sum^{N/2}_{j=0} \sum^{j-1}_{m_j = -j} \big[ \cos(J^+_{j,m}t) \mathbf{1} + \sin(J^+_{j,m}t) \hat \sigma^x  \big] \mathscr{P}_j,\\
 \label{u_t}
\end{eqnarray}

where $m_j$ is the magnetic quantum number corresponding to the azimuthal quantum number $j$,
\begin{eqnarray}
 J_{j,m}^{+} =  \sqrt{(j - m) ( j + m + 1)},
\label{eq: J}  
\end{eqnarray}
 is an eigenvalue of the raising operator. $\mathbf{1}$  denotes the identity operator and $\sigma^{\alpha} \; \alpha \in \lbrace x,y,z\rbrace $ denote the respective Pauli operators on the basis of the $2 \times 2$ block diagonals of $H^{res}$. The projection operator $\mathscr{P}_j$ corresponds to each bath-spin $j$ sector.

   The chosen half-cycle division between (5a) and (5b) is simple \emph{but not unique:}   one could assume the time intervals corresponding to the swap and the dephasing operations to be unequal in each cycle (say, $\tau^{res}_n$ and $\tau^{disp}_n$ dividing the $n$'th cycle) so as to \emph{minimize the number of cycles} required to reach the desired state, i.e., maximize the purification speed. In that case, however, the problem becomes intractable, even numerically, as it demands simultaneous optimization over an infinite number of parameters $(\tau_1, \tau_2,...)$. Hence, we assume, for the sake of simplicity, that each cycle is halved between the resonant swap and the off-resonant dephasing operations and the cycle durations $\tau$ are all the same. These $\tau$ are assumed to be of the order of the inverse average value of the spin-chain couplings; i.e., $\tau \sim (\overline{g})^{-1}$, to maximize the excitation transfer from $S$ to $A$.

\setcounter{equation}{0}
\renewcommand{\theequation}{E\arabic{equation}}
\section{Solvable models of a spin chain   attached to an ancilla spin}\label{appendix:E}

These are expressed by the collective spin operators
\begin{eqnarray}
&& \hat{J}^\alpha =  \frac{1}{2} \sum^N_{k=1} \hat{\sigma}^\alpha_k; ~~~, \alpha = x,y,z \nonumber\\ 
&& \hat{J}^2 = (\sum^N_{k=1} \hat{\sigma}_k)^2.
\end{eqnarray}
The $SU(2)$ symmetry of the joint $(S+A)$ Hamiltonian under RT (Eq.~(\ref{drt_a})) is confirmed  by the  commutation relations
\begin{eqnarray}
 &&[\hat{H}_S, \hat{J}^\alpha] = 0; ~~~ \forall \alpha = x,y,z.\nonumber\\
 &&[\hat{H}_S, \hat{J}^2] = [\hat{H}_{SA}^{res}, \hat{H}_S] = 0.
\end{eqnarray}

These commutation relations imply that the eigenstates of the operators $\hat{J}^2$ and $\hat{J}^z$ are the same as the eigenstates of $\hat{H}_S$. Using the Bethe-Anstaz~\cite{bethe,izyumov} one can obtain the highest-energy degenerate eigenstates from the eigenstates of the raising $\hat{J}^+ = \hat{J}^x + i\hat{J}^y$ operators. Progressively lower-energy eigenstates are obtained by applying the lowering $\hat{J}^- = \hat{J}^x - i\hat{J}^y$ operator to these highest-energy Bethe states, thus yielding the complete basis of $\hat{H}_S$ eigenstates (SI.~\ref{appendix:D}). 

The basis states of the full Hamiltonian satisfy the following relations
\begin{eqnarray}
&&(\hat H^{res}_{SA}+\hat H_S)\vert \pm \frac{1}{2};\alpha,j,m =  \pm j\rangle = E_{\alpha,j}\vert \pm \frac{1}{2};\alpha,j,m=\pm j\rangle\nonumber\\ 
&&(\hat H^{res}_{SA}+\hat H_S)\vert\frac{1}{2};\alpha,j,m\rangle = E_{\alpha,j}\vert\frac{1}{2};\alpha,j,m\rangle  +J^+_{j,m} \vert-\frac{1}{2};\alpha,j,m+1\rangle\nonumber\\ 
&&(\hat H^{res}_{SA}+\hat H_S)\vert-\frac{1}{2};\alpha,j,m+1\rangle = E_{\alpha,j}\vert-\frac{1}{2};\alpha,j,m+1\rangle +J^+_{j,m} \vert\frac{1}{2};\alpha,j,m\rangle,\nonumber\\
\end{eqnarray}
with $J^+_{j,m}$ as in (\ref{eq: J}). 
From the above equations we find that the states $\vert\frac{1}{2};\alpha,j,m = j\rangle$ and $\vert-\frac{1}{2};\alpha,j,m = -j\rangle$ are eigenstates of the joint system and have energies $E_{\alpha,j}$.
As a consequence, the Hilbert space of the system can be divided into a direct sum of $2\times 2$ blocks of the following form with basis states $\vert\frac{1}{2};\alpha,j,m\rangle$ and $\vert-\vert\frac{1}{2};\alpha,j,m+1\rangle$.

\begin{eqnarray}
 \begin{bmatrix}
E_{\alpha,j} & J^+_{j,m}\\
J^+_{j,m} & E_{\alpha,j}
\end{bmatrix}   
\end{eqnarray}
where $m\ne j$. The eigenstates and eigenvalues are given as
\begin{eqnarray}
   \epsilon^{\pm}_{\alpha,j,m} = E_{\alpha,j} \pm J^+_{j,m}.
\end{eqnarray}

\begin{eqnarray}
 \vert \epsilon^{\pm}_{\alpha,j,m}\rangle   =  \frac{1}{\sqrt{2}}\vert\frac{1}{2};\alpha,j,m\rangle \pm \frac{1}{\sqrt{2}}  \vert-\frac{1}{2};\alpha,j,m+1\rangle. 
\end{eqnarray}

The limit of  $N$ non-interacting/isolated spins (star model) is obtained as their interaction strength $\mathcal{J}$ in Eq.~(1b) is much smaller than the $S-A$ couplings $g_k$. In this limit,  the eigenvalues simplify to $\epsilon^{\pm}_{j,m_j} =  \pm J^+_{j,m_j}$. 

The time evolution operator is thus given as 
\begin{eqnarray}
\hat U(t) = \sum_{\alpha,j,m} \sum_{\pm}e^{-it\epsilon^{\pm}_{\alpha,j,m}} \vert \epsilon^{\pm}_{\alpha,j}\rangle  \langle \epsilon^{\pm}_{\alpha,j}\vert.
\end{eqnarray}

It is natural to represent the dynamics of a system of non-interacting spins with an ancilla-spin in the irreducible or collective basis states. For even and odd numbers of qubits, the number of distinct $j$ values are $(\frac{N}{2} +1)$ and $(\frac{N+1}{2})$ respectively. In the case of $SU(2)$, combining $N$ doublets yields the Clebsch-Gordan decomposition series with the multiplicity of all possible representations~\cite{zachos_1992} 

\begin{eqnarray}
  \mathbf{2}^{\otimes N} = \bigoplus^{[N/2]}_{k=0} \left(\frac{N+1-2k}{N+1} \binom{N+1}{k}\right) (\mathbf{N}+\mathbf{1} - \mathbf{2k} ), \nonumber\\
  \label{young_tableau}
\end{eqnarray}
 where $[N/2]$ is the integer floor function. The degeneracy with each $j$ value occurs in the basis states in the irreducible representation is obtained by setting $N+1- 2k =  2j+1$ in Eq.~\ref{young_tableau},

 \begin{eqnarray}
d_j = \frac{2j+1}{N+1} \binom{N+1}{\frac{N}{2} - j}.
 \end{eqnarray}
 
The set of the complete basis for $N$ qubits can be uniquely represented by the quantum numbers corresponding to the set of observables $\{J_z, J_1,...,J_{N-1}\}$ with $J_k \equiv (\sum^{k+1}_{i=1} I_i)^2$. The quantum numbers corresponding to the same observables can be found using the method of Young Tableau~\cite{sakurai_book}. In order to reach the fully polarized state, the map $\mathbcal{F}$   must not commute with anyone from this set of observables, and hence the observables corresponding to these operators must not be conserved throughout the dynamics.


\setcounter{equation}{0}
\renewcommand{\theequation}{F\arabic{equation}}

\section{RT and ADRT in solvable models}\label{appendix:F}

\noindent \textbf{With RT:}

The state of $S+A$ in the collective angular momentum basis, after many resets of $A$, assumes the mixed diagonal form~\cite{alvarez_2010}
\begin{eqnarray}
 \rho^{(n)}_{S} =  \sum^{N/2}_{j=0(\frac{1}{2})} \sum^{j}_{m = -j}\sum^{d_j}_{\alpha =1} p^{(n)}_{\alpha,j,m} |\alpha, j, m\rangle \langle \alpha_j, j, m|, \nonumber\\
 \label{eq: drt_st}   
\end{eqnarray}

where $p^{(n)}_{\alpha,j,m}$ the probability of the $S$ eigenstate $|\alpha,j, m\rangle$ after the $n'~$th reset, is determined by the $H_{SA}$ action.

 The reset of $A$ reshuffles the blocks corresponding to different collective angular momenta $|\frac{1}{2}; \alpha, j,m\rangle$ and $|-\frac{1}{2}; \alpha,j,m+1\rangle$. Each block with particular $(j,m)$ values retains a certain fraction of the probability $p^{(n)}_{\alpha,j,m}$  and transfers the rest to the block corresponding to $(m+1)$ with the same of $j$.The block with the lowest $m$  value gains nothing, while the block with the highest value of $m = j$ accumulates all probabilities coming from the lower $m$ blocks. At the $n'~$th cycle, these probabilities $p^{(n)}_{\alpha, j,m}$ satisfy the recurrence relations~\cite{alvarez_2010}

 \begin{eqnarray}
   p^{(n+1)}_{\alpha,j,m} = p^{(n)}_{\alpha,j,m} \cos^2(J^+_{j,m} \tau) +  p^{(n)}_{\alpha,j,m-1} \sin^2(J^+_{j,m -1} \tau). \nonumber\\
   \label{eq: recurrence}
 \end{eqnarray}

  As we can see from the set of recurrence relations in Eq.~(\ref{eq: recurrence}),  a flow of probabilities $p^{(n)}_{\alpha, j,m}$ takes place within each $j$ sector from lower to higher values of $m_j$. After $n \rightarrow \infty$ cycles,  the probabilities accumulate at the highest $m = j$ blocks of each $j$ value.  Hence, the invariant blocks with $(j,m=j)$ cause a bottleneck effect that impedes further purification. The minimum value of the entropy of $S$ attainable 

 \begin{eqnarray}
\lim_{n\rightarrow \infty} \mathcal{S}^{(n)}_S = - \sum^{N/2}_{j=0(\frac{1}{2})} \sum^{d_j}_{\alpha=1} p^{(n)}_{\alpha,j,m=j} \ln p^{(n)}_{\alpha,j,m=j} \nonumber\\
\label{eq:DRT_system}
 \end{eqnarray}

The final values of the purity are given by the following forms,

\begin{eqnarray}
  P &=& \sum^{N/2}_{j=0} d_j \left(\frac{2j+1}{2^N}\right)^2 = \frac{1}{(2^N)^2 (N+1)} \sum^{N/2}_{j=0} (2j+1)^3 \binom{N+1}{\frac{N}{2} - j}  = \frac{(N/2 +1)(2N+1)}{4^N (N+1)} \binom{N+1}{\frac{N}{2}}, {~~~}\text{even~}N, \nonumber\\ 
    P &=& \sum^{N/2}_{j=1/2} d_j \left(\frac{2j+1}{2^N}\right)^2 = \frac{1}{(2^N)^2 (N+1)} \sum^{N/2}_{j=1/2} (2j+1)^3 \binom{N+1}{\frac{N}{2} - j}  = \frac{(N+3)}{4^N} \binom{N+1}{\frac{N-1}{2}}, {~~~}\text{odd~}N. \nonumber\\ 
\end{eqnarray}

\noindent \textbf{With ADRT:}

The full protocol of ADRT, consecutively consists of the resetting of $A$, polarization swapping between $S$ and $A$, and dephasing of the $S-A$ compound (by their dispersive coupling),  breaks all the existing symmetries in $S+A$. The state of $S$ for the general case after $n+1$ cycles is generally solved by the recurrence relation as a function of its state after $n$ cycles $\rho^{(n)}_S = \sum_{\bold{x}_1(l_1),\bold{x}_2(l_2)} \rho^{(n)}_{\bold{x}_1(l_1),\bold{x}_2(l_2)} \vert \bold{x}_1(l_1) \rangle \langle \bold{x}_2(l_2) \vert$  as
\begin{subequations}
\begin{eqnarray}
&&\rho^{(n+1)}_{\bold{x}_1(l_1),\bold{x}_2(l_2)} = \sum_{\bold{y}_1(l_1),\bold{y}_2(l_2)}\mathcal{G}^{0,\bold{x}_1(l_1)}_{0,\bold{y}_1(l_1)}(\tau) \mathcal{G}^{*0,\bold{x}_2(l_2)}_{0,\bold{y}_2(l_2)}(\tau) \rho^{(n)}_{\bold{y}_1(l_1),\bold{y}_2(l_2)} \nonumber\\
&& +\sum_{\bold{y}_1(l_1+1),\bold{y}_2(l_2+1)}\mathcal{G}^{1,\bold{x}_1(l_1)}_{0,\bold{y}_1(l_1+1)}(\tau) \mathcal{G}^{*1,\bold{x}_2(l_2)}_{0,\bold{y}_2(l_2+1)}(\tau) \rho^{(n)}_{\bold{y}_1(l_1+1),\bold{y}_2(l_2+1)}. \nonumber\\
\label{eq: propagator}
\end{eqnarray}
where,  $(s_A, \bold{x}(l) \equiv (x_1,x_2,...,x_l))$ are the computational basis representation of the joint $S-A$ state, involving the propagators
\begin{eqnarray}
\mathcal{G}^{s'_A,\bold{y}_1(l_2)}_{s_A,\bold{x}_1(l_1)}(\tau) \equiv \langle s'_A,\bold{y}_1 (l_2) | T_{\leftarrow} \text{e}^{-i\int^\tau_0 \hat{H}(t) dt}|s_A,\bold{x}_1 (l_1)\rangle.
\label{propagator}
\end{eqnarray}    
\end{subequations}

However, for solvable models, the action of the time evolution operator $\hat U^{res}_{SA}(t)$ on the computational basis states

  \begin{eqnarray}
 &&e^{-it \hat H^{res}_{SA}}|0;\bold{x}(l)\rangle  = \cos(J^+_{jm}t) \vert0,\bold{x}(l)\rangle - i\sum_{\bold{y}(l-1)} \psi^{\bold{y}(l-1)}_{\bold{x}(l)} \sin(J^+_{jm}t) \vert1,\bold{y}(l-1)\rangle, \nonumber\\
 \end{eqnarray}

where $\psi^{\bold{y}(l-1)}_{\bold{x}(l)} \equiv C^{\alpha,j,m}_{\bold{x}(l)} C^{\bold{y}(l-1)}_{\alpha,j,m+1}$ is the wavefunction corresponding to the transition between the states $\bold{x}(l)$ and $\bold{y}(l-1)$ which is a function of the Clebsh-Gordan coefficients relating the computational-basis and the collective-basis states. The azimuthal quantum number $m$ and the magnon number $l$ are related as $l+m=N/2$. 

Similarly, The action of the time evolution operator $\hat U^{disp}_{SA}(t)$ on the computational basis states

  \begin{eqnarray}
 &&e^{-it \hat H^{disp}_{SA}}|0;\bold{x}(l)\rangle  = e^{-it\phi(\bold{x}(l))} \vert0,\bold{x}(l)\rangle, \nonumber\\
  &&e^{-it \hat H^{disp}_{SA}}|1;\bold{x}(l)\rangle  = e^{it\phi(\bold{x}(l))} \vert1,\bold{x}(l)\rangle, \nonumber\\
 \end{eqnarray}

where, $\phi(\bold{x}(l)) \equiv \big( \sum^N_{k=1} g_k (-1)^{\sgn{l_k}}   \big)$. Therefore, the propagators are 

\begin{eqnarray}
    &&\mathcal{G}^{0,\bold{y}(l)}_{0,\bold{x}(l)}(t) = \cos(J^+_{jm}t)  e^{-it\phi(\bold{x}(l))} \delta_{\bold{y}(l),\bold{x}(l)},\nonumber\\
     &&\mathcal{G}^{1,\bold{y}(l-1)}_{0,\bold{x}(l)}(t) = -i \sin(J^+_{jm}t) \psi^{\bold{y}(l-1)}_{\bold{x}(l)} e^{-it\phi(\bold{y}(l-1))}.
\end{eqnarray}

\section{Circuit-QED implementation of non-reciprocal Networks}
Graphs with nonreciprocal couplings~\cite{RanzaniAumentado2015} may affect the purity bounds discussed in the manuscript. In the presence of nonreciprocal couplings, we enter the category of directed graphs that are inherently asymmetric. One possible implementation of such a partially nonreciprocal network can be 
realized in a superconducting circuit-QED platform consisting of four qubits 
$Q_1, Q_2, Q_3, Q_4$ interconnected through parametrically modulated couplers. Reciprocal interactions between most qubit pairs can be generated using standard tunable couplers operated with symmetric parametric modulation \cite{McKay2016, PhysRevLett.127.080505}. 
To engineer nonreciprocity exclusively between $Q_1$ and $Q_2$, one may introduce a parametrically driven three-wave–mixing element that imparts a synthetic gauge phase to the exchange interaction, thereby breaking reciprocity while leaving the rest of the network reciprocal 
\cite{Kamal2011, RanzaniAumentado2015, Xu2020, Chapman2017}. 
Such schemes allow embedding a directional coupling within an otherwise 
reciprocal multi-qubit architecture.

\begin{figure}
 \centering
 \includegraphics[width=0.5\linewidth]{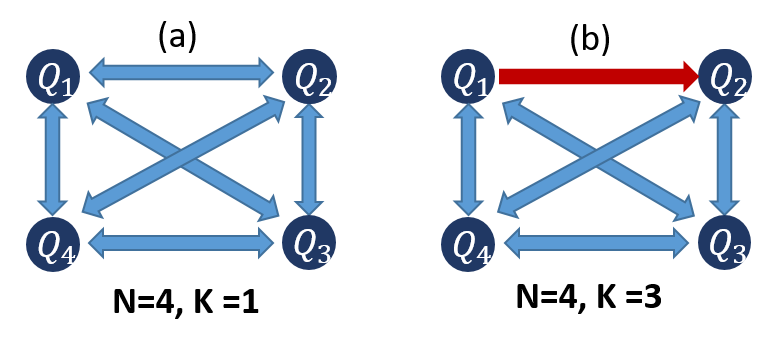}
  \caption{{(a) Schematic of a simple symmetric graph with all-to-all reciprocal couplings that has only one AO ($K=1$) and (b) another graph with increased asymmetry and AO's ($K=3$) when one of the vertices is made directed, i.e., coupling between $Q_1 - Q_2$ is non-reciprocal. } }
  \label{nonrecgraph}
  \end{figure}

Schematically, the qubit connectivity is illustrated in Fig.~\ref{nonrecgraph}. Here, the arrows indicate directional (nonreciprocal) coupling, while the bidirectional links represent standard reciprocal interactions.  In addition to the four primary qubits, we introduce an ancilla qubit \(R\) that can be periodically reset to its ground state \(|0\rangle\). This qubit $R$ is coupled, say, to \(Q_3\), providing a mechanism to \emph{purify or stabilize} the state of the four-qubit system. The initial condition for the experiment assumes that \(Q_1, Q_2, Q_3, Q_4\) are all prepared in the excited state \(|1\rangle\), while the ancilla \(R\) is periodically reset throughout the protocol \(|0\rangle\). This setup allows us to study combined unitary and dissipative dynamics in a small qubit network.
The Hermitian Hamiltonian describing the qubit energies and reciprocal couplings reads:
\begin{align}
H_Q &= \sum_{j=1}^{4} \frac{\omega_j}{2} \sigma_j^z + \frac{\omega_R}{2} \sigma_R^z, \\
H_{\rm rec} &= \sum_{\substack{i<j \\ (i,j)\neq (1,2)}} J_{ij} (\sigma_i^+ \sigma_j^- + \sigma_i^- \sigma_j^+), \\
H_{\rm ancilla} &= J_{R3} (\sigma_R^+ \sigma_3^- + \sigma_R^- \sigma_3^+),
\end{align}
where \(\omega_j\) are qubit transition frequencies, \(J_{ij}\) are reciprocal coupling strengths, and \(J_{R3}\) is the ancilla-qubit coupling.

To implement a non-reciprocal (directional) interaction between \(Q_1\) and \(Q_2\), we couple both qubits to a  waveguide mode \(a\) :

\begin{equation}
H_{\rm wg} = g_1 (\sigma_1^+ a + \sigma_1^- a^\dagger) + g_2 (\sigma_2^+ a + \sigma_2^- a^\dagger) + \omega_a a^\dagger a,
\end{equation}
where \(g_1, g_2\) are coupling strengths to the waveguide and \(\omega_a\) is its mode frequency. This Hamiltonian is Hermitian, as required. By tracing out the waveguide degrees of freedom under the Born-Markov approximation, we obtain an  effective directional (chiral) interaction  between \(Q_1\) and \(Q_2\) in Lindblad form \cite{PhysRevResearch.6.L032017, MetelmannClerk2015}:

\begin{equation}
\mathcal{L}_{12}[\rho] = \Gamma_{12} \Big( \sigma_2^- \rho \, \sigma_1^+ - \frac{1}{2} \{\sigma_1^+ \sigma_2^-, \rho\} \Big),
\end{equation}
where \(\Gamma_{12}\) sets the effective  rate of directional excitation transfer  from \(Q_1 \to Q_2\).  

This framework captures the all-to-all reciprocal couplings and the waveguide-mediated nonreciprocal interaction, and the resettable ancilla dynamics, forming the basis for studying graphs with non-reciprocal couplings. With increased asymmetry and the AO's, such graphs may be better polarized than highly symmetric graphs. 
Such effects are highly relevant for spin transport studies, where the nonreciprocity of spintronic devices may have a strong impact on the device performance~\cite{Yu2023ChiralitySpintronics,Geng2023}.

\end{document}